\documentclass[reprint,
superscriptaddress,
 prx,amsmath,amssymb,
 aps]{revtex4-2}
 
\usepackage[normalem]{ulem}
\usepackage{appendix}
\usepackage{blindtext}
\usepackage{multirow}    
\usepackage{graphicx}
\usepackage{mathtools}
\usepackage{dcolumn}
\usepackage{amsthm}
\usepackage{amsmath}
\usepackage{amsfonts}
\usepackage{amssymb}
\usepackage{bm}
\usepackage{bbm}
\usepackage{color}
\usepackage[dvipsnames]{xcolor}

\DeclareMathOperator{\Tr}{Tr}
\theoremstyle{definition}
\newtheorem{theorem}{Theorem}

\newtheorem{claim}{Claim}
\newtheorem{fact}{Fact}

\theoremstyle{remark}

\newcommand{\XZ}{X\! \! Z}

\newcommand{\rev}[1]{\textcolor{Red}{#1}}
\newcommand{\egr}{\textcolor{ForestGreen}} 

\renewcommand{\egr}{\textcolor{Black}} 
\renewcommand{\rev}[1]{\textcolor{Black}{#1}}

\begin{document}

\title{Character randomized benchmarking for non-multiplicity-free groups with applications to subspace, leakage, and matchgate randomized benchmarking}

\author{Jahan Claes}
\affiliation{Quantum Artificial Intelligence Laboratory (QuAIL), NASA Ames Research Center, Moffett Field, CA 94035, USA}
\affiliation{USRA (RIACS), Mountain View CA 94043, USA}
\affiliation{Department of Physics and Institute for Condensed Matter Theory, University of Illinois at Urbana-Champaign, Urbana, IL 61801, USA}

\author{Eleanor Rieffel}
\affiliation{Quantum Artificial Intelligence Laboratory (QuAIL), NASA Ames Research Center, Moffett Field, CA 94035, USA}

\author{Zhihui Wang}
\email{zhihui.wang@nasa.gov}
\affiliation{Quantum Artificial Intelligence Laboratory (QuAIL), NASA Ames Research Center, Moffett Field, CA 94035, USA}
\affiliation{USRA (RIACS), Mountain View CA 94043, USA}

\begin{abstract}
    Randomized benchmarking (RB) is a powerful method for determining the error rate of experimental quantum gates. 
    Traditional RB, however, is restricted to gatesets, such as the Clifford group, that form a unitary 2-design. The recently introduced character RB can benchmark more general gates using techniques from representation theory; up to now, however, this method has only been applied to ``multiplicity-free" groups, a mathematical restriction on these groups. In this paper, we extend the original character RB derivation to explicitly treat non-multiplicity-free groups, and derive several applications. First, we derive a rigorous version of the recently introduced subspace RB, which seeks to characterize a set of one- and two-qubit gates that are symmetric under SWAP. 
    Second, we develop a new leakage RB protocol that applies to more general groups of gates. Finally, we derive a scalable RB protocol for the matchgate group, a group that like the Clifford group is non-universal but becomes universal with the addition of one additional gate. This example provides one of the few examples of a scalable non-Clifford RB protocol.
    In all three cases, compared to existing theories, our method requires similar resources, but either provides a more accurate estimate of gate fidelity, or applies to a more general group of gates.
    In conclusion, we discuss the potential, and challenges, of using non-multiplicity-free character RB to develop new classes of scalable RB protocols and methods of characterizing specific gates.
\end{abstract}

\date{\today}

\maketitle

\section{Introduction}

Advances in accurate and scalable methods for characterizing the performance of quantum gates are critical for the realization of large-scale reliable quantum computers. Quantum process tomography can, in theory, completely characterize an unknown quantum channel \cite{o2004quantum,chuang1997prescription,poyatos1997complete,mohseni2008quantum}, but requires resources that scale exponentially in the number of qubits \cite{mohseni2008quantum}. In addition, any tomographic approach will also include the effect of state preparation and measurement (SPAM) errors, which may be of the same order as the gate error that is being characterized.

Randomized benchmarking (RB) \cite{emerson2005scalable,knill2008randomized,magesan2011scalable,magesan2012characterizing} provides a method to scalably characterize gates that form a group $G$ with the additional mathematical property of being a ``unitary 2-design" \cite{dankert2009exact}, most frequently the Clifford group \cite{gottesman1998theory,gottesman1998heisenberg,aaronson2004improved}. Rather than completely characterizing a noise channel, RB determines the average fidelity, a standard measure of gate quality that can be related to other common measures such as entanglement and process fidelity \cite{nielsen2002simple,horodecki1999general} and used to bound the gate error rate \cite{sanders2015bounding}. RB works by experimentally measuring the overall fidelity of a random circuit as a function of the number of applied gates $U\in G$ and fitting this to an exponential decay. The parameters of the decay then determine the average fidelity of a single gate. Unlike tomographic methods, RB provides an estimate for the average fidelity that is independent of SPAM errors.

Standard RB, however, is limited to groups that form a unitary 2-design and whose elements can be efficiently compiled (i.e. decomposed) into elementary gates. This limitation prevents standard RB from characterizing any set of quantum gates that are large enough to be universal for quantum computation \cite{gottesman1998heisenberg,aaronson2004improved}, and also prevents standard RB from characterizing smaller subgroups of 2-designs. There are ongoing efforts to extend RB to a larger class of gates. Interleaved RB was proposed to characterize individual Clifford group elements \cite{magesan2012efficient} as well as the $T$-gates needed for universal quantum computation \cite{harper2017estimating}, but these methods are specific to the gates considered and only produce bounds on the fidelity. Ref. \cite{carignan2015characterizing} developed a method to extract the fidelity of the dihedral group on one qubit, which is not a unitary 2-design and includes the $T$ gate, while \cite{cross2016scalable} proposed a method of extending dihedral RB to an arbitrary number of qubits. Refs. \cite{brown2018randomized,francca2018approximate} extended this work by deriving decay formulas for the fidelity of random circuits of arbitrary groups, but these formulas involved fitting sums of multiple exponentials, and the decay parameters could not be related to the average fidelity. Ref. \cite{helsen2019new} introduced character RB to address these limitations, providing a method that only requires fitting a single exponential decay and directly predicts the average fidelity. However, this was only explored for ``multiplicity-free" groups, a mathematical limitation on the group's representations (see below).

In this work, we provide a generalization of character RB that applies to \rev{groups with multiplicity,} \egr{which we underpin with rigorous derivations. This rigor enables us to provide conditions under which instantiations of the framework yield practical RB protocols.} We illustrate our generalized approach with applications to three distinct situations of practical interest: benchmarking of gates with subspace preserving properties, characterization of leakage, and benchmarking of the matchgate group.

\egr{Our main contributions include:}
\begin{itemize}
 \item \rev{We provide a derivation of character RB for non-multiplicity-free groups $G$. This} RB method allows us to directly predict the average fidelity of the gates in $G$ as in \cite{helsen2019new} but unlike \cite{brown2018randomized,francca2018approximate}. For non-multiplicity-free groups, our method potentially requires fitting a sum of multiple exponentials rather than a single exponential; however, the number of exponentials is significantly reduced compared to \cite{brown2018randomized,francca2018approximate}. 
\item \rev{As a} primary motivation for this generalization, \rev{we} improve the recently introduced subspace RB \cite{baldwin2020subspace} designed to characterize gates that preserve a subspace of the full Hilbert space. \egr{Our generalization, and its rigorous derivations, has immediate application to near-term quantum processors, including to benchmarking the gates implemented on the ion-trap quantum processor benchmarked in \cite{baldwin2020subspace}. Gates that preserve a proper subspace} can never form a 2-design, and are never multiplicity-free, necessitating a generalized RB procedure. The original work on subspace RB established decay formulas for the fidelity of certain random circuits but could only give loose bounds on the average fidelity of the gates; our method, in contrast, allows us to directly estimate the average fidelity using a similar number of experiments as the original subspace RB. \egr{While we illustrate our approach for the $U_{ZZ}$ gate seen in \cite{baldwin2020subspace}, the method can be applied directly to other gates with the same SWAP symmetry as the $U_{ZZ}$ gate. It also provides grounding for benchmarking gates with other subspace-preserving symmetries, though creativity will be required to determine when and how these gates can be combined with single qubit gates to obtain a group with the properties that yield a practical character RB protocol. The rigorous derivations underlying our approach enables us to provide examples of noise under which the \rev{estimated fidelity} yielded by \cite{baldwin2020subspace} deviates substantially from the exact fidelity \rev{provided by our method}.}
 \item 
We present a new protocol for leakage RB \cite{chasseur2015complete,wallman2016robust,wood2018quantification}, a benchmarking protocol designed to characterize qubits that can ``leak" into a non-computational section of the Hilbert space. \egr{Our approach reduces the assumptions on control in the leakage subspace required by the original leakage RB work \cite{wood2018quantification}. Such control is frequently unrealistic for quantum hardware. Our approach can be applied immediately to determine certain leakage channel error rates in, for example, quantum dot architectures, though further research will need to be done to obtain a leakage RB protocol that enables the determination of more general parameters including the average fidelity on the computational subspace.} 
\item 
We introduce a new scalable RB procedure for the matchgate group \cite{valiant2002quantum}, a class of quantum circuits that, like the Clifford group, is efficiently simulable \cite{valiant2002quantum,brod2016efficient,terhal2002classical,jozsa2008matchgates} but is very close to universal \cite{jozsa2008matchgates,brod2014computational,imamog1999quantum,terhal2002classical,brod2011extending,brod2012geometries,hebenstreit2020computational}. This procedure necessarily requires the full non-multiplicity-free character RB, and represents, along with the dihedral group \cite{cross2016scalable,helsen2019new}, one of the few non-Clifford groups that can be scalably benchmarked.
\end{itemize}

Non-multiplicity-free character RB is a general framework for benchmarking groups \egr{of quantum gates}. It provides a method for characterizing individual gates when the \rev{gates can be combined into} operations that form a group, \rev{as we illustrate in the case of subspace RB.} \rev{This RB framework also} expands the family of groups that can be scalably benchmarked, \rev{as we demonstrate with the matchgate group}. Scalable benchmarking protocols are necessary to measure gate quality in large quantum processors, \rev{especially in the presence of non-local errors such as crosstalk}. While we provide one \rev{new} example of a scalable benchmarking protocol, we expect the framework \rev{of} non-multiplicity-free character RB will lead researchers to develop further scalable examples. Benchmarking multiple overlapping groups (or subgroups of groups) may allow more accurate error characterization.
\egr{While it remains an art to find the groups and constructions that yield practical character RB protocols, we expect the grounding that our work provides to support the discovery of practical protocols for various gate sets in a variety of quantum devices in the years to come.}

Our paper is organized as follows. Section \ref{section:mathematicalPreliminaries} provides mathematical background on the Liouville representation and the definition of average fidelity. Section \ref{section:generalizedRB} outlines the full non-multiplicity-free RB protocol, and proves that it correctly estimates the average fidelity of the gates. The next sections consist of applications. Section \ref{section:subspace} demonstrates how our method can be used to rigorously estimate the fidelity of gate sets that preserve subspaces, such as those studied in \cite{baldwin2020subspace}. Section \ref{section:leakage} applies our framework to formulate a leakage RB protocol with fewer assumptions than the current state-of-the-art \cite{wood2018quantification}. Section \ref{section:matchgate} reviews the matchgate group, and describes how our method can be used to derive a scalable RB protocol for this group. \rev{Each of our applications are accompanied by computer simulations of benchmarking experiments; all our computer simulations can be reproduced in under a day on a standard laptop.} We conclude in Section \ref{section:conclusion} with discussion of possible extensions of our work, including some of the challenges. We relegate technical details to appendices, including \rev{Appendix \ref{appendix:gateDependent} which demonstrates that our method is robust to gate-dependent errors}, and Appendix \ref{appendix:WeylGroup} which provides a self-contained and straightforward proof that generalizations of the Clifford group to qudits for $d$ prime form a unitary $2$-design, which may be of \rev{independent} interest.

\section{Mathematical Preliminaries}
\label{section:mathematicalPreliminaries}
In this paper, we use the Liouville representation of quantum channels. In the Liouville representation, given some fixed basis $\{|i\rangle\}$ of our Hilbert space $\mathcal H$, a density matrix $\rho=\sum_{ij}\rho_{ij}|i\rangle\langle j|$ is represented by a column vector $|\rho\rangle\rangle=\sum_{ij}\rho_{ij}|i\rangle \otimes|j\rangle$, where we use a double-bracket $|\cdot\rangle\rangle$ to distinguish elements of $\mathcal{H}\otimes\mathcal{H}$ from elements of $\mathcal{H}$. In the case of a pure state $\rho=|\psi\rangle\langle\psi|$ we will also sometimes write $|\psi\rangle\rangle$ in place of $|\rho\rangle\rangle$. 
A quantum channel $\Lambda(\rho)=\sum_i A_i\rho A_i^\dagger$ is represented by a matrix $\hat{\Lambda}=\sum_{i}A_i\otimes A_i^*$. In this representation, matrix multiplication corresponds to composition
$$\widehat{\Lambda_1\circ\Lambda_2}=\hat{\Lambda}_1\hat{\Lambda}_2,$$ matrix-vector multiplication corresponds to applying a quantum channel $$\hat{\Lambda}|\rho\rangle\rangle = |\Lambda(\rho)\rangle\rangle,$$ and the inner product of two vectors corresponds to the Hilbert-Schmidt inner product of the corresponding density matrices $$\langle\langle\sigma|\rho\rangle\rangle = \Tr(\sigma^\dagger \rho).$$ In particular, if $M$ is a projector into some measurement outcome, the overlap $\langle\langle M|\rho\rangle\rangle$ gives the probability of measuring $M$ from a state $\rho$. For a more detailed treatment of the Liouville representation, see \cite{wood2015tensor}.

Given a unitary group $G$ acting on our Hilbert space $\mathcal H$, the natural action of $U\in G$ on density matrices is given by $U(\rho)=U\rho U^\dagger$. In the Liouville representation, such an operator is represented by $\hat{U}=U\otimes U^*$. The map $\phi: U\mapsto U\otimes U^*$ forms a representation \cite{fulton2013representation} of the group $G$ on $\mathcal{H}\otimes\mathcal{H}$ that we will refer to as the {\bf natural representation} of $G$. We can also define the {\bf $\bm{G}$-twirl} of a quantum channel $\Lambda$ as 
\begin{equation}
    \hat{\Lambda}_G = \frac{1}{|G|}\sum_{U\in G} \hat{U}^\dagger\hat{\Lambda}\hat{U}.
    \label{eq:twirl}
\end{equation}
where $|G|$ is the order of the group. We can also define the $G$-twirl by compact groups by replacing the discrete average by the integral over the Haar measure. As we will see, $\Lambda_G$ has properties similar to the original channel $\Lambda$, but it has a simpler structure that makes it more tractable to study.

If a noisy implementation of a gate $U$ results in applying the channel $(\Lambda\circ U)$, we want to characterize how close the noise channel $\Lambda$ is to the identity. We will focus on one common measure of noise, the {\bf average fidelity} $F_\Lambda$, given by 
\begin{equation}
    F_\Lambda :=\int d\psi \langle\langle\psi|\hat\Lambda|\psi\rangle\rangle.
\end{equation}
Here, $d\psi$ is the unitary-invariant Haar or Fubini-Study measure on $\mathcal{H}$. The integrand $\langle\langle\psi|\hat{\Lambda}|\psi\rangle\rangle$ is the probability of preserving a state $|\psi\rangle$ after the noise operator $\Lambda$ has been applied. The average fidelity is then simply the average of this probability over all possible input states.

\section{The generalized character randomized benchmarking procedure}
\label{section:generalizedRB}
Let $G$ be the unitary group on $\mathcal{H}$ that we wish to benchmark. We will assume $G$ is either finite or compact, so that every unitary representation decomposes into irredicible representations. Let $\phi:G\rightarrow \mathcal{L}(\mathcal{H}\otimes\mathcal{H})$ be the natural representation of $G$, which decomposes into irreducible representations as $\phi\simeq a_1\phi_1\oplus\cdots\oplus a_I\phi_I$, where $a_i\in\mathbb{Z}^+$ is the multiplicity of the irrep $\phi_i$. Let $\mathcal{H}\otimes\mathcal{H}\simeq \bigoplus_i \mathbbm{C}^{a_i}\otimes \mathcal{H}_i$ be the corresponding decomposition of Hilbert space, such that each $\phi_i$ acts nontrivially only on a single copy of $\mathcal{H}_i$. We will make the standard RB assumption that the gate error $\Lambda$ associated with $U\in G$ is independent of $U$, although this can be relaxed \cite{proctor2017randomized,wallman2018randomized,merkel2018randomized,helsen2019new}(see Appendix \ref{appendix:gateDependent}). 

Let $\overline G\subseteq G$ be a subgroup of our unitary group with natural representation $\overline\phi \simeq \overline{a}_1\overline{\phi}_1\oplus\cdots\oplus \overline{a}_{\overline{I}}\overline{\phi}_{\overline{I}}$ and corresponding decomposition $\mathcal{H}\otimes\mathcal{H}\simeq \bigoplus_i \mathbbm{C}^{\overline{a}_{\overline i}}\otimes \overline{\mathcal{H}}_{\overline i}$. We choose $\overline G$ such that for every $i\in\{1,...,I\}$, there exists a corresponding $\overline{i}\in \{1,...,\overline I\}$ such that $\mathbbm{C}^{\overline{a}_{\overline i}}\otimes \overline{\mathcal{H}}_{\overline i}\subseteq \mathbbm{C}^{a_i}\otimes\mathcal{H}_i$. One may satisfy this condition by choosing $\overline G=G$, but we will see below that for this procedure to scale with the number of qubits we must choose $\overline G\subsetneq G$. We denote the character of the irrep $\overline{\phi}_{\overline i}$ by $\chi_{\overline i}(U):=\text{Tr}\left[{\overline\phi}_{\overline i}(U)\right]$.

Our RB procedure consists of the following steps:
\begin{enumerate}
    \item For each $i\in\{1,...,I\}$, choose an initial state $|\rho_i\rangle\rangle$ and measurement projector $|M_i\rangle\rangle$ such that $|\langle\langle M_i|\hat{P}_{\overline i}|\rho_i\rangle\rangle|$ is large as possible (see Section \ref{section:scaling} below), where $\hat{P}_{\overline i}$ is the projector onto $\overline{\mathcal{H}}_{\overline i}$.
    \item For a given $N$, choose unitaries $U_0\in \overline G$ and $U_1,...,U_N\in G$ randomly and uniformly (note elements can be repeated). \rev{In the case of a compact group rather than a finite group, choose elements according to the Haar measure.} Compute $U_{N+1}=U_1^\dagger\cdots U_N^\dagger$.
    \item Prepare the state $|\rho_i\rangle\rangle$. Apply the gates $(U_1U_0), U_2,..., U_{N+1}$ sequentially, where $({U}_1{U}_0)$ is compiled as a single element of $G$.
    \item Perform a measurement of the observable \rev{$M_i$}.
    \item Repeat steps 2-4 many times, to estimate the {\bf character-weighted survival probability} 
    \rev{\begin{equation}
        S_i(N)=\frac{1}{|G|^{N+1}}\sum_{\substack{U_0\in \overline{G}\\U_1,...,U_N\in G}} \chi^*_{\overline i}(U_0)\text{Pr}_{U_0,...,U_{N+1}}
        \label{eq:survival}
    \end{equation}}\noindent 
    for each $i$, where $\text{Pr}_{U_0,...,U_{N+1}}$ is the probability of measuring $|M_i\rangle\rangle$ after applying gates \rev{$(U_1U_0),...,U_{N+1}$} to \rev{$|\rho_i \rangle\rangle$}, including the effect of gate and SPAM errors.
    \item Repeat steps 2-5 for different values of $N$.
    \item Fit each character-weighted survival probability to a function of the form
    \begin{equation}
        S_i(N) = \sum_{j=1}^{a_i} C_{i,j}\lambda_{i,j}^N
        \label{eq:fitForm}
    \end{equation}
    where the $C_{i,j}$ and $\lambda_{i,j}$ are \rev{(possibly complex)} fitting parameters independent of $N$. \rev{Note that if $\chi_{\overline{i}}$ is complex we may have $S_i$ complex, but if $\chi_{\overline{i}}$ is real the $C_{i,j}$ and $\lambda_{i,j}$ are restricted to be real or come in complex-conjugate pairs}.
    \item Estimate the average fidelity of the gate error $\Lambda$ as
    \begin{equation}
        F_\Lambda = \frac{\sum_{i=1}^I\left[\text{dim}(\mathcal{H}_i)\sum_{j=1}^{a_i}\lambda_{i,j}\right]+d}{d^2+d}
        \label{eq:fidelityEst}
    \end{equation}
    where $d:=2^n$ is the dimension of Hilbert space.
\end{enumerate}

A similar RB procedure was first proposed in \cite{helsen2019new} for groups with all $a_i=1$, the so-called {\bf multiplicity-free} groups. In this case, each character-weighted survival probability becomes a single exponential decay. Character RB had been previously proposed for the multiplicity-free dihedral group on one qubit \cite{carignan2015characterizing}, and a related approach has been used to simplify standard RB \cite{harper2019statistical}.

We note if we omit the initial gate $U_0$ and the character-weighting $\chi_{\overline i}^*(U_0)$, we get the method of \cite{cross2016scalable,brown2018randomized,francca2018approximate}; in this case, we get a \emph{single} survival probability $S(N)$ that is given by $S(N)=\sum_{i,j} C_{i,j}\lambda_{i,j}^N$. Determining the $\lambda_{i,j}$ then requires fitting all the parameters $C_{i,j}$ and $\lambda_{i,j}$ simultaneously, and quickly becomes infeasible for a modestly large number of parameters. We see that while both our method and the method of \cite{cross2016scalable,brown2018randomized,francca2018approximate} involve simultaneously fitting multiple exponential decays, our method significantly reduces the number of parameters in each fit. For example, if $\phi\simeq 2\phi_1\oplus\phi_2\oplus\phi_3$, our method requires fitting three functions, corresponding to $\phi_1$, $\phi_2$, and $\phi_3$, where the first function is a sum of two exponential decays and the latter two functions are single exponential decays. In contrast, \cite{cross2016scalable,brown2018randomized,francca2018approximate} require fitting a single exponential function that is the sum of four exponential decays, one for each copy of each irrep. In addition, the method of \cite{cross2016scalable,brown2018randomized,francca2018approximate} cannot determine $F_\Lambda$, because it is not possible to match the observed parameters $\{\lambda_{i,j}\}$ to their corresponding $\mathcal{H}_i$ in order to use Eq. \ref{eq:fidelityEst}.

The remainder of this section is devoted to deriving this procedure, for groups that are not necessarily multiplicity-free. Much of this is a straightforward extension of the derivation of \cite{helsen2019new}, although the generalization to gate-dependent noise (Appendix \ref{appendix:gateDependent}) is much less straightforward.

\subsection{Deriving the decays}
\label{section:decayDerivation}
To derive the form of the character-weighted survival, Eq. \ref{eq:fitForm}, we will need two facts from representation theory.

\begin{fact}[Schur's Lemma]
    Let $\phi:G\rightarrow \mathcal{L}(V)$ be a representation of a group $G$ on a vector space $V$, which decomposes into irreducible representations as $\phi\simeq a_1\phi_1\oplus\cdots\oplus a_I\phi_I$, where $a_i\in\mathbb{Z}^+$ are positive integers. The corresponding decomposition of $V$ is $V\simeq \bigoplus_i \mathbbm{C}^{a_i}\otimes V_i$. In terms of this decomposition, any linear map $\hat{\eta}\in\mathcal{L}(V)$ satisfying $\hat\eta\phi(U)=\phi(U)\hat\eta$ for all $U\in G$ is of the form 
    \begin{equation}
        \hat\eta \simeq  \bigoplus_i \hat{Q}_i\otimes \hat{\mathbbm{1}}_i
    \end{equation}
    where $\hat Q_i$ is some $a_i\times a_i$ matrix for each $i$.
    \label{fact:Schur}
\end{fact}
\begin{fact}[Projection formula]
    Let $\phi$ and $V$ be as above. Given an irrep $\phi_i:G\rightarrow \mathcal{L}(V_i)$, define the {\bf character} $\chi_i:G\rightarrow \mathbbm{C}$ of $\phi_i$ as $\chi_i(U):=\Tr\left(\phi_i(U)\right)$. Then we can write the projector onto $\mathbbm{C}^{a_i}\otimes V_i$ as
    \begin{equation}
        \hat{P}_i=\frac{\text{dim} (V_i)}{|G|}\sum_{U\in G}\chi_i(U)^*\phi(U).
        \label{eq:Projection}
    \end{equation}
    \label{fact:Projection}
\end{fact}\noindent
For proofs of both facts, see \cite{fulton2013representation}.

Given these results, we can prove the key property of $G$-twirls that allows us to compute the average fidelity.
\begin{theorem}[Form of $G$-twirls]
    If $G$ is any unitary group acting on $\mathcal{H}$, let $\phi\simeq a_1\phi_1\oplus\cdots\oplus a_I\phi_I$ be the decomposition of the natural representation into irreps, and let $\mathcal{H}\otimes\mathcal{H}\simeq \bigoplus_i \mathbbm{C}^{a_i}\otimes \mathcal{H}_i$ be the corresponding decomposition of $\mathcal{H}\otimes\mathcal{H}$. If $\Lambda$ is any quantum channel, the $G$-twirl of $\Lambda$ is of the form
    \begin{equation}
        \hat{\Lambda}_G \simeq \bigoplus_i \hat{Q}_i\otimes \hat{\mathbbm{1}}_i\;,
        \label{eq:twirlForm}
    \end{equation}
    where $Q_i$ is defined as in Fact. \ref{fact:Schur}.
    \label{thm:twirlForm}
\end{theorem}
\begin{proof}
    We apply Eq. \ref{eq:twirl} to observe that 
    \begin{align*}
        \hat{\Lambda}_G \hat{U} &= \frac{1}{|G|}\sum_{U'\in G} \hat{U}'^\dagger\hat{\Lambda}\hat{U}'\hat{U}\\
        &=\frac{1}{|G|}\sum_{U'\in G}\hat{U}\hat{U}^\dagger \hat{U}'^\dagger\hat{\Lambda}\hat{U'}\hat{U}\\
        &=\hat{U}\frac{1}{|G|}\sum_{(U'U)\in G}(\hat{U}'\hat{U})^\dagger\hat{\Lambda}(\hat{U}'\hat{U})=\hat{U}\hat{\Lambda}_G
    \end{align*}
    for any $U\in G$. We can then apply Fact \ref{fact:Schur}.
\end{proof}

We are now ready to derive the formula for the character-weighted survival probability $S_i(N)$. This proof follows the logic of \cite{helsen2019new}, adapted for non-multiplicity-free groups. \rev{Our notation assumes finite groups; for compact groups, one simply replaces the discrete average over the group with an integral over the Haar measure.} Writing out Eq. \ref{eq:survival} explicitly, including the effect of preparation and measurement errors $\Lambda_P$ and $\Lambda_M$, we have
\begin{widetext}
\begin{align*}
    {S}_i(N)&=\frac{1}{|G|^{N}|\overline{G}|}\sum_{U_0,...,U_N}\underbracket{\chi_{\overline i}^*(U_0)}_{\hat{P}_i}\langle\langle M_i|\hat{\Lambda}_M\hat\Lambda \hat{U}_{N+1}\hat\Lambda\hat{U}_N\cdots\hat\Lambda\hat{U}_2\hat\Lambda\hat{U}_1\underbracket{\hat{U}_0}_{\hat{P}_{\overline i}}\hat{\Lambda}_P|\rho_i\rangle\rangle.
\end{align*}
The sum over $U_0$ gives the projection $|\overline G|\hat{P}_{\overline i}/\text{dim}(\overline{\mathcal{H}}_{\overline i})$ according to Eq. \ref{eq:Projection}. To do the sum over $U_1,...,U_N$, we can define new group elements $D_1,...,D_N$ by $D_i=U_i\cdots U_1$. In terms of the $D_i$, we then have $U_i=D_iD_{i-1}^\dagger$, with the convention that $D_{N+1}=\mathbbm{1}$. Note that summing over $U_1,...,U_N$ is the same as summing over $D_1,...,D_N$. We therefore may write
\begin{align*}
    {S}_i(N)&=\frac{1}{\text{dim}(\overline{\mathcal{H}}_{\overline i})|G|^{N}}\sum_{D_1,...,D_N\in G}\langle\langle M_i|\hat{\Lambda}_M\hat\Lambda\underbrace{ \hat{D}_N^\dagger\hat\Lambda\hat{D}_N}_{\hat\Lambda_G}\cdots\underbrace{\hat{D}_2^\dagger\hat{\Lambda}\hat{D}_2}_{\hat\Lambda_G}\underbrace{\hat{D}_1^\dagger\hat\Lambda\hat{D}_1}_{\hat\Lambda_G}\hat{P}_{\overline i}\hat{\Lambda}_P|\rho_i\rangle\rangle.
\end{align*}
\end{widetext}
We can now easily perform the sum over the $D_i$, since each sum just gives a $G$-twirl according to Eq. \ref{eq:twirl}. Performing this sum, and using Thm. \ref{thm:twirlForm}, gives
\begin{align*}
    S_i(N)&=\frac{1}{\text{dim}(\mathcal{\overline H}_{\overline i})}\langle\langle M_i|\hat{\Lambda}_M\hat\Lambda \left(\hat{\Lambda}_G\right)^N\hat{P}_{\overline i}\hat{\Lambda}_P|\rho_i\rangle\rangle\\
    &=\frac{1}{\text{dim}(\mathcal{\overline H}_{\overline i})}\langle\langle M_i|\hat{\Lambda}_M\hat\Lambda \left(\bigoplus_{i'}\hat{Q}_{i'}\otimes \mathbbm{1}_{i'}\right)^N\hat{P}_{\overline i}\hat{\Lambda}_P|\rho_i\rangle\rangle\\
    &=\frac{1}{\text{dim}(\mathcal{\overline H}_{\overline i})}\langle\langle M_i|\hat{\Lambda}_M\hat\Lambda \left(\hat{Q}_{i}^N\otimes \mathbbm{1}_{i}\right)\hat{P}_{\overline i}\hat{\Lambda}_P|\rho_i\rangle\rangle\\
\end{align*}
where in the last line, we used the fact that the range of $\hat{P}_{\overline i}$ is included in $\mathbbm{C}^{a_i}\otimes\mathcal{H}_i$. We see that the effect of the character-weighting is to produce a projector that restricts our attention to a single $i$. If we diagonalize $\hat{Q}_i$ as $\hat{Q_i}=\sum_{j=1}^{a_i}|e_{i,j}\rangle\rangle \lambda_{i,j}\langle\langle \overline{e}_{i,j}|$ with $\langle\langle \overline{e}_{i,j}|{e}_{i,j'}\rangle\rangle=\delta_{j,j'}$, then $\hat{Q}_i^N=\sum_{j=1}^{a_i}|e_{i,j}\rangle\rangle \lambda_{i,j}^N\langle\langle \overline{e}_{i,j}|$, and we may write the final form of $S_i(N)$ as
\begin{equation*}
    S_i(N) = \sum_{j=1}^{a_i}\frac{\langle\langle M_i|\hat{\Lambda}_M\hat\Lambda \Big{(}|e_{i,j}\rangle\rangle\langle\langle \overline{e}_{i,j}|\otimes \mathbbm{1}_{i}\Big{)}\hat{P}_{\overline i}\hat{\Lambda}_P|\rho_i\rangle\rangle}{\text{dim}(\overline{\mathcal H}_{\overline i})}\lambda_{i,j}^N
\end{equation*}
which is precisely the form given in Eq. \ref{eq:fitForm}. Notice that the $\lambda_{i,j}$ depend only on the gate error $\Lambda$, and not the SPAM errors $\Lambda_P,\Lambda_M$ which are absorbed into the constant prefactor.

\subsection{Computing the fidelity}

Finally, we prove the fidelity can be estimated according to Eq. \ref{eq:fidelityEst}. This was first derived in \cite{francca2018approximate}, although we will adopt a simpler proof here using techniques introduced in \cite{nielsen2002simple,horodecki1999general}. The key realization is that \emph{both} the fidelity and the trace of a channel are invariant under twirling by an arbitrary group: $F_\Lambda=F_{\Lambda_G}$ and $\Tr(\hat\Lambda)=\Tr(\hat\Lambda_G)$ (see Eq. \ref{eq:twirl}). In particular, if we choose $G$ to be the full unitary group it is known that the full twirl of a channel is simply a depolarizing channel \cite{horodecki1999general,nielsen2002simple}\footnote{In our notation, this can be seen by noting that the natural representation of the full unitary group decomposes into two irreps which act on $|\mathbbm{1}\rangle\rangle$ and the orthogonal complement of $|\mathbbm{1}\rangle\rangle$, respectively, and then applying Fact \ref{fact:Schur}.}:
\begin{equation}
    \hat{\Lambda}_G:=\int dU\ \hat{U}^\dagger\hat{\Lambda}\hat U = p\mathbbm{1}+(1-p)\frac{1}{d}|\mathbbm{1}\rangle\rangle\langle\langle\mathbbm{1}|.
    \label{eq:twirlDepolarizing}
\end{equation}
In terms of the parameter $p$, we can directly compute $F_{\Lambda_G}=p+\frac{1-p}{d}$. Similarly, we can also directly compute $\Tr(\hat{\Lambda}_G)=pd^2+(1-p)$. Combining these equations gives
\begin{equation}
F_\Lambda = \frac{\Tr(\hat\Lambda)+d}{d^2+d}.
\label{eq:TraceFidelity}
\end{equation}
To complete the proof, we note that $\Tr(\hat{\Lambda})$ can be written in terms of the matrices $\hat{Q}_i$ in Eq. \ref{eq:twirlForm} as $$\Tr(\hat{\Lambda})=\sum_{i=1}^I\left[\text{dim}(\mathcal{H}_i)\Tr(\hat{Q}_i)\right]=\sum_{i=1}^I\left[\text{dim}(\mathcal{H}_i)\sum_{j=1}^{a_i}\lambda_{i,j}\right]$$
which, combined with Eq. \ref{eq:TraceFidelity}, gives Eq. \ref{eq:fidelityEst} as desired.

\subsection{Scaling and Feasibility}
\label{section:scaling}
We note that experimentally determining $S_i(N)$ requires Monte Carlo sampling of $U_0,U_1,...,U_N$. Each term in this sample is bounded by $\max_{U_0\in \overline G}(|\chi_{\overline i}(U_0)|)=\text{dim}({\mathcal{H}}_{\overline i})$. Therefore, the standard deviation of the samples is bounded by $\text{dim}({\mathcal{H}}_{\overline i})$, and the sample mean has uncertainty bounded by $\text{dim}({\mathcal{H}}_{\overline i})/\sqrt{\text{no. samples}}$. To determine the \emph{relative} uncertainty, we consider $S_i(N)\approx \sum_{j=1}^{a_i}C_{i,j}$ which is given by
\begin{align*}
\sum_{j=1}^{a_i}C_{i,j} &= \sum_{j=1}^{a_i}\frac{\langle\langle M_i|\hat{\Lambda}_M\hat\Lambda \Big{(}|e_{i,j}\rangle\rangle\langle\langle \overline{e}_{i,j}|\otimes \mathbbm{1}_{i}\Big{)}\hat{P}_{\overline i}\hat{\Lambda}_P|\rho_i\rangle\rangle}{\text{dim}(\mathcal{\overline H}_{\overline i})} \\
&\approx\frac{\langle\langle M_i|\hat{P}_{\overline i}|\rho_i\rangle\rangle}{\text{dim}(\mathcal{\overline H}_{\overline i})}
\end{align*}
where we've approximated $\Lambda,\Lambda_M,\Lambda_P\approx \mathbbm{1}$. The relative uncertainty in $S_i(N)$ is therefore bounded by
$$
\frac{\sigma_i}{|S_i(N)|}\lesssim\frac{\text{dim}(\overline{\mathcal{H}}_{\overline i})^2}{|\langle\langle M_i|\hat{P}_{\overline i}|\rho_i\rangle\rangle|\sqrt{\text{no. samples}}}
$$
We see that to efficiently benchmarking a group $G$, we must have $I$, $a_i$, and $\dim(\overline{\mathcal{H}}_{\overline i})$ all small. $I$ must be small so that we only need to estimate a small number of character-weighted survival probabilities $S_i(N)$, $a_i$ must be small so that we may fit a function with a small number of parameters, and $\dim(\overline{\mathcal{H}}_{\overline i})$ must be small for our Monte Carlo estimation of $S_i(N)$ to converge quickly. Note that for any $G$ the natural representation satisfies $\sum_{i=1}^I a_i\dim(\mathcal{H}_i)=4^{n}$ where $n$ is the number of qubits, so that choosing $\overline G = G$ will not suffice if the number of qubits is large. In particular, to {\bf scalably} benchmark a group, we must choose $G$ so that the number of irreps $I$ grows slowly with $n$, the multiplicity $a_i$ of each irrep is bounded by a small constant, and $\overline G$ has corresponding irreps $\overline{\mathcal{H}}_{\overline i}$ whose dimension grows slowly with $n$. These scaling considerations are similar to those discussed in \cite{helsen2019new} for multiplicity-free RB, except in our case we allow $a_i$ to be bounded rather than strictly $1$.

Note that the optimal $|\rho_i\rangle\rangle$ with largest $|\langle\langle M_i|\hat{P}_{\overline i}|\rho_i\rangle\rangle|$ is necessarily a pure state, since any mixed state $|\rho_i\rangle\rangle=\sum_\gamma p_\gamma |\psi_\gamma\rangle\rangle$ has
$$
|\langle\langle M_i|\hat{P}_{\overline i}|\rho_i\rangle\rangle|\leq \sum_\gamma p_\gamma|\langle\langle M_i|\hat{P}_{\overline i}|\psi_\gamma\rangle\rangle|\leq \max_\gamma |\langle\langle M_i|\hat{P}_{\overline i}|\psi_\gamma\rangle\rangle|.
$$
Ref. \cite{helsen2019new} considered the case of mixed initial states, and included a protocol for sampling from a mixed state $|\rho_i\rangle\rangle=\sum_\gamma p_\gamma|\psi_\gamma\rangle\rangle$ provided one can efficiently prepare the states $\{|\psi_\gamma\rangle\rangle\}$. However, we see that it suffices to take the initial state to be one of the efficiently preparable $|\psi_\gamma\rangle\rangle$, which simplifies initial state preparation.

Our scaling estimates are based on the typical case; however, there are a few worst-case failure modes. First, the noise may have some symmetry that restricts $\langle\langle \overline{e}_{i,j}|\hat{P}_{\bar i}\approx 0$ for some $(i,j)$. In this case, the corresponding $\lambda_{i,j}$ will not be accurately estimated by the fitting function. To remedy this, one may choose a set of projectors $\left\{\hat{P}_{\overline i,1},...,\hat{P}_{\overline i,k}\right\}$ such that each $\langle\langle \overline{e}_{i,j}|$ has overlap with at least one $\hat{P}_{\overline i,\alpha}$. This requires at most $a_i$ projectors. We can then define
$$\hat{P}_{\overline i}=\sum_\alpha\hat{P}_{\overline i, \alpha}\quad \chi_{\overline i}=\sum_\alpha\chi_{\overline i, \alpha}.$$
The modified character-weighted survival probability will require taking additional data to achieve the same relative uncertainty, since the corresponding $\dim(\overline{\mathcal{H}}_{\overline i})=\sum_\alpha \dim(\overline{\mathcal{H}}_{\overline i,\alpha})$ will be larger, but is otherwise identical.

The fitting procedure may also have difficulty fitting multiple exponential decays \rev{\cite{bromage1983quantification,clayden1992multiexponential}}, especially if the decay rates are similar \rev{\cite{clayden1992multiexponential}}. In the case of similar decays, the fit might have numerous local minima; worse, the fitting function might simply set the coefficient of one of the decays to zero and the corresponding decay rate to some arbitrary value, and fit the curve using fewer exponential decays. This can be detected during the fitting procedure, and corrected by either taking more data to more closely constrain the fit or by simply fitting fewer exponential decays. \rev{For a detailed discussion of methods used to fit multiexponential decays and their failure modes, we refer to \cite{istratov1999exponential,holmstrom2002review,hokanson2013numerically}.}

\section{Application: Subspace randomized benchmarking}

\label{section:subspace}

As an application of the general character RB method, we can improve on the recently introduced subspace randomized benchmarking method \cite{baldwin2020subspace}. Subspace RB characterizes the error associated with a group of gates $G$ that preserve a subspace of the Hilbert space. In \cite{baldwin2020subspace}, a benchmarking procedure is introduced that yields two decay parameters that are functions of the noise channel, but the procedure does not give an estimate for the average fidelity or other quantities with simple physical interpretations. The multiplicity-free character RB of \cite{helsen2019new} is not directly applicable to this situation, as we will see that any group that preserves subspaces necessarily decomposes into irreps with multiplicity. However, using our method we can easily characterize the average fidelity of such gates.

To simplify our discussion, we will focus on the particular case discussed in \cite{baldwin2020subspace}. The system considered in \cite{baldwin2020subspace} can implement arbitrary symmetric single qubit gates $U_1:=U\otimes U$ as well as the two-qubit entangling gate $U_{ZZ}:=\exp\{-i\frac{\pi}{4}Z\otimes Z\}$. The symmetric single qubit gates have negligible error compared to the entangling gate, so the goal of the experiment is to characterize the fidelity of $U_{ZZ}$. This is accomplished by combining the elementary gates into elements of a benchmarking group $G$, using a fixed number of the relevant gate $U_{ZZ}$, and then designing an RB procedure to benchmark elements of $G$. It is straightforward to see that any $U\in G$ made up of products of $U_1$ and $U_{ZZ}$ operators preserves the triplet and singlet subspaces
\begin{align*}
\mathcal{H}_T &:= \text{span}\left\{|00\rangle,\frac{|01\rangle+|10\rangle}{\sqrt{2}},|11\rangle\right\}\\
\mathcal{H}_S &:= \text{span}\left\{\frac{|01\rangle-|10\rangle}{\sqrt{2}}\right\}.
\end{align*}
This implies that every gate $U\in G$ decomposes as $U=U_T\oplus U_S$, with $U_T$ and $U_S$ acting on the triplet and singlet spaces, respectively.

Our method differs from the original in several ways. Most notably, we combine the elementary gates into elements $U\in G$ such that $G$ forms a group. This requires a moderate increase in complexity of the combined gates; \cite{baldwin2020subspace} combined their gates into unitaries involving three $U_{ZZ}$ gates, while our construction requires four. However, in return for this increased complexity, our method offers several advantages. Rather than estimate decay parameters with no clear physical interpretation, our method produces direct estimates of the average fidelity. In addition, the derivation of the form of the exponential decays in \cite{baldwin2020subspace} required assumptions on the relative phases of $U_T$ and $U_S$ that could not actually be realized on their experimental platform. In contrast, our method yields rigorous decays thanks to the underlying group structure of $G$.

The original subspace RB can be extended to sets of gates $G$ that preserve some arbitrary splitting of $\mathcal{H}$ into subspaces $\mathcal{H}=\mathcal{H}_1\oplus\mathcal{H}_2$ provided the set $G$ can be written as
\rev{$$
G=\{U_{1,b_1}\oplus \sigma U_{2,b_2}: \sigma=\pm,\ (b_1,b_2)\in B_1\times B_2\}
$$}
where \rev{$G_1:=\{U_{1,b_1}:b_1\in B_1\}$} and \rev{$G_2:=\{U_{2,b_2}:b_2\in B_2\}$} are groups and unitary 2-designs \footnote{Ref. \cite{baldwin2020subspace} claimed it was sufficient to require \rev{$G_2$} to be a unitary $1$-design, but this appears to be an error. A similar error was made in \cite{wood2018quantification}, from which much of \cite{baldwin2020subspace} is derived.} (see below for the definition of a $2$-design) \rev{acting on $\mathcal{H}_1$ and $\mathcal{H}_2$ respectively (here, $B_1$ and $B_2$ are just index sets for the groups $G_1$ and $G_2$)}. However, it is difficult to construct such a \rev{$G$} in a way that is experimentally relevant; indeed, \cite{baldwin2020subspace} could not do this for the simple case of two qubits, and we avoid attempting such a construction here. A more useful approach, which mirrors our approach below, is to construct an arbitrary group out of the elementary gates and perform character RB on whatever irreps result. This method can likely be used to benchmark other two-qubit gates that are symmetric under SWAP besides $U_{ZZ}$, and may also prove useful for gates that preserve other subspaces.

\subsection{Constructing the benchmarking group}

Ref. \cite{baldwin2020subspace} constructed their benchmarking set $G$ using a generalization of the Clifford group \cite{gottesman1998heisenberg,gottesman1998heisenberg,aaronson2004improved} to a $d$-level system \cite{hostens2005stabilizer}. We will follow a similar procedure, modified to ensure $G$ forms a group. For a $d$-level system, analogues of the $X$ and $Z$ qubit operators are defined as \cite{gottesman1998fault}:
$$
X|z\rangle = |z+1\rangle \qquad Z|z\rangle = \omega^z|z\rangle
$$
where $\omega := e^{\frac{2\pi i}{d}}$ and addition is performed modulo $d$. These generalized $X$ and $Z$ operators are unitary \rev{but not Hermitian}, and the set $\{X^aZ^b:a,b\in\mathbbm{Z}_d\}$ forms a (complex) orthogonal basis for the set of all $d\times d$ matrices. Note that for $d=2$ we recover the usual Pauli matrices.

Specializing to $d=3$, define the generalized Pauli group as $\mathcal{P}:=\{\omega^\eta X^a Z^b:\eta,a,b\in\mathbbm{Z}_d\}$. The fact that $\mathcal{P}$ is a group follows from the commutation relation $ZX = \omega XZ$. The generalized Clifford group is defined to be the set of all unitaries that stabilize $\mathcal{P}$ \cite{hostens2005stabilizer}:
$$
G_T = \{U :U\mathcal{P} U^\dagger = \mathcal{P}\}.
$$
An element $U\in G_T$ is defined (up to a global phase) by its action on $X$ and $Z$. Defining $UXU^\dagger=\omega^{\eta_x}X^{a_x}Z^{b_x}$ and $UZU^\dagger=\omega^{\eta_z}X^{a_z}Z^{b_z}$, and noting
\begin{align*}
ZX&=\omega XZ\\
UZU^\dagger UXU^\dagger&=\omega UXU^\dagger UZU^\dagger\\
\omega^{\eta_x+\eta_z}X^{a_z}Z^{b_z}X^{a_x}Z^{b_x}&=\omega^{1+\eta_x+\eta_z}X^{a_x}Z^{b_x}X^{a_z}Z^{b_z}\\
\omega^{a_xb_z}X^{a_x+a_z}Z^{b_x+b_z}&=\omega^{1+a_zb_x}X^{a_x+a_z}Z^{b_x+b_z}
\end{align*}
we see that we must have $a_xb_z-a_zb_x=_3 1$, where $=_3$ denotes equality mod $3$. This is the only restriction on $\eta_x,\eta_z,a_x,a_z,b_x,b_z$ \cite{hostens2005stabilizer}, leading to a total of 216 elements of $G_T$. We can find the action of $U\in G_T$ on a general element $X^aZ^b$ by
\begin{align*}
\begin{split}
    U X^aZ^b U^\dagger = &\ (U XU^\dagger)^a(UZU^\dagger)^b\\
    = &\ \omega^{P} X^{aa_x+ba_z}Z^{ab_x+bb_z}
\end{split}
\end{align*}
where
$$P:=\eta_xa+\eta_zb+2(a^2-a)a_xb_x+2(b^2-b)a_zb_z+abb_xa_z.$$
The action of $U$ on a general density matrix then follows by linearity.

Our benchmarking group $G$ is constructed by combining the elementary symmetric gates to act as $G_T$ on the triplet subspace, where the three levels $|0\rangle,|1\rangle,|2\rangle$ correspond to the triplet basis $|00\rangle,\frac{|01\rangle+|10\rangle}{\sqrt 2},|11\rangle$. The most general composite gate is formed by alternatively applying $U_1$ and $U_{ZZ}$ gates to our qubits. A straightforward calculation shows that if such a circuit applies an operator $U_T$ to the triplet subspace, its action on the singlet subspace is necessarily given by $(-1)^{n_z}\omega^\eta\det(U_T)^{1/3}$, where $n_z$ is the number of entangling $U_{ZZ}$ gates. By varying the single-qubit unitaries $U_1$, we find computationally that all elements of $G_T$ and all relative phases $\omega^\eta$ can be generated by circuits of exactly four $U_{ZZ}$ gates, as shown in Fig. \ref{fig:SubspaceCircuit} \footnote{Ref. \cite{baldwin2020subspace} required a shorter circuit of only three entangling gates. However, this circuit cannot implement all relative phases between the subspaces and thus does not result in a group.}. In total, then, the benchmarking group is given by
$$
G:=\{U_T\oplus \omega^\eta \det(U_T)^{1/3} : U_T\in G_T, \eta=0,1,2\}
$$
where the first summand acts on the triplet subspace and the second acts on the singlet subspace. Note that every group element contains exactly four entangling gates, so the average fidelity of $G$ gives a useful measure of the fidelity of the entangling gate.

\begin{figure}
    \centering
    \includegraphics[width=\columnwidth]{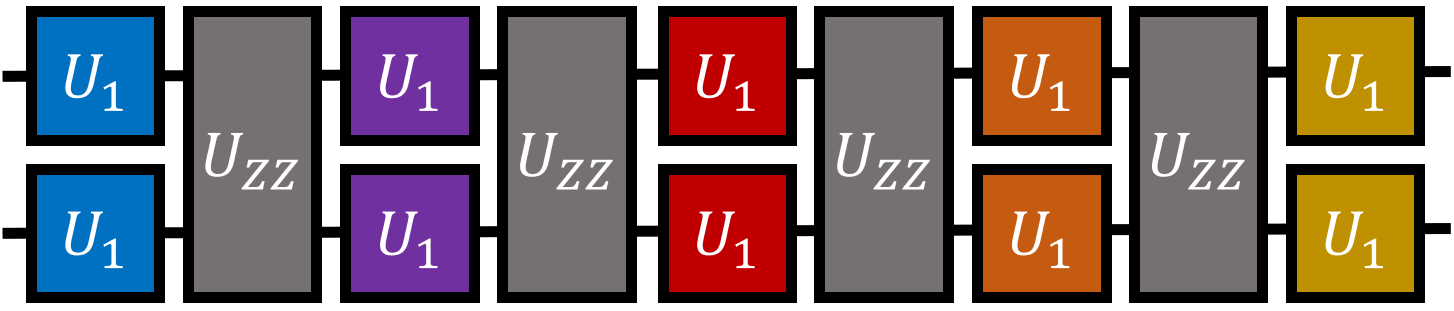}
    \caption{The elements of the benchmarking group $G$ are constructed by composing elementary gates as shown above to implement elements of $G_T$ on the triplet subspace. Each group element contains exactly four entangling gates.}
    \label{fig:SubspaceCircuit}
\end{figure}

\renewcommand{\arraystretch}{1.5}
\begin{table}
    \centering
    \begin{tabular}{c|c|c}
        Subrep & Projector & $\chi_i(U_T\oplus U_S)$ \\\hline\hline
        $\mathcal{H}_{T0}$  & $\hat{P}_{T0}=\frac{1}{3}|\mathbbm{1}_T\rangle\rangle\langle\langle\mathbbm{1}_T|$&\multirow{2}{*}{$1$} \\\cline{1-2}
        $\mathcal{H}_{S0}$ & $\hat{P}_{S0}=|\mathbbm{1}_S\rangle\rangle\langle\langle\mathbbm{1}_S|$& \\\hline
        $\mathcal{H}_{T\perp}$ & $\hat{P}_{T\perp}=\mathbbm{1}_T-\hat{P}_{T0}$ & $|\Tr(U_T)|^2-1$ \\\hline
        $\mathcal{H}_{TS}$ & $\hat{P}_{TS}=\text{Projector onto }\mathcal{H}_{T}\otimes \mathcal{H}_{S}$ & $\Tr(U_T)\Tr(U_S)^*$ \\\hline
        $\mathcal{H}_{ST}$ & $\hat{P}_{ST}=\text{Projector onto }\mathcal{H}_{S}\otimes \mathcal{H}_{T}$ &  $\Tr(U_T)^*\Tr(U_S)$ \\\hline
    \end{tabular}
    \caption{Subrepresentations of the standard representation for groups that preserve the triplet and singlet subspaces, and their corresponding projectors and characters.}
    \label{table:characters}
\end{table}

\subsection{Irreps of the benchmarking group}
\label{section:irrepsSubspace}
For $G$ given above, the natural representation decomposes into the irreps $\mathcal{H}_{T0}$, $\mathcal{H}_{S0}$, $\mathcal{H}_{T\perp}$, $\mathcal{H}_{TS}$, and $\mathcal{H}_{ST}$, which are described in Table \ref{table:characters}. These are all clearly subrepresentations of the natural representation; for proof that they are in fact irreducible, we will use the concept of a {\bf unitary $\mathbf{t}$-design} \cite{dankert2009exact}.

Let $S$ be a set of unitaries acting on a space $\mathcal H$. A {\bf balanced polynomial} of degree $t$ is a polynomial in the matrix elements of $U$ and $U^*$ where each term in the polynomial has degree $d<t$ in the elements of $U$ and degree $d$ in the elements of $U^*$. $S$ is a unitary $t$-design if for balanced polynomial  $p(U,U^*)$ of degree $t$, averaging $p(U,U^*)$ over $S$ is the same as averaging over all unitaries on $\mathcal H$ (weighted by the Haar measure)
$$
\frac{1}{|S|}\sum_{U\in S} p(U,U^*)=\int dU\ p(U,U^*).
$$
A classic example is the Clifford group, which forms a unitary $3$-design \cite{dankert2009exact,webb2016clifford,zhu2017multiqubit}.

The group $G_T$ forms a unitary 2-design \cite{chau2005unconditionally} (see Appendix \ref{appendix:WeylGroup} for a proof). This allows us to prove the representations in Table \ref{table:characters} are irreducible, using the following fact:
\begin{fact}[Schur normalization]
Let $\chi$ be the character of a representation. The representation is irreducible iff $$\frac{1}{|G|}\sum_{U\in G}|\chi(U)|^2=1.$$\label{fact:magnitudeCharacter}
\end{fact}\noindent
For a proof, see \cite{fulton2013representation}.

The representations $\mathcal{H}_{T 0}$ and $\mathcal{H}_{S 0}$ are 1D, thus irreducible. For the representation $\mathcal{H}_{T\perp}$, we have
\begin{align*}
\frac{1}{|G|}\sum_{U\in G}|\chi_{T\perp}(U)|^2&=\frac{1}{3|G_T|}\sum_{\substack{U_T\in G_T\\\eta=0,1,2}}|\chi_{T\perp}(U_T)|^2\\
&=\frac{1}{|G_T|}\sum_{G_T}\left(|\Tr(U_T)|^2-1\right)^2\\
&=\int dU_\alpha\ \left(|\Tr(U_\alpha)|^2-1\right)^2 \\
&= 1
\end{align*}
where the second equality follows from the unitary 2-design property, and the third follows from the fact that $\mathcal{H}_{T\perp}$ is an irrep of the natural representation of the full unitary group on $\mathcal{H}_T$. Finally, for $\mathcal{H}_{TS}$ and $\mathcal{H}_{ST}$ we have
\begin{align*}
    \frac{1}{|G|}\sum_{U\in G}|\chi_{ST}(U)|^2&=\frac{1}{3|G_T|}\sum_{\substack{U_T\in G_T\\\eta=0,1,2}}|\Tr(U_T)|^2\\
    &=\int dU_T\ |\Tr(U_T)|^2\\
    &=1
\end{align*}
where the second equality follows from the unitary 2-design property and the third follows from the fact that the direct representation of the full unitary group on $\mathcal{H}_T$ is irreducible.

Note that $\mathcal{H}_{T0}$ and $\mathcal{H}_{S0}$ are two irreducible copies of the trivial representation, so that $G$ is necessarily non-multiplicity-free \footnote{It follows that $G$ also cannot form a 2-design, as 2-designs are always multiplicity free; in particular, the natural representation of a 2-design decomposes into precisely two non-isomorphic irreps, acting on $|\mathbbm{1}\rangle\rangle$ and the orthogonal complement of $|\mathbbm{1}\rangle\rangle$ \cite{dankert2009exact,helsen2019new}.}. The remaining irreps are all unique, since they have different character functions. 

\subsection{Benchmarking $G$}

The form of the decay curves corresponding to each irrep is given by
\begin{equation}
\begin{aligned}
    S_0(N)&= C_0\lambda_0^N+B\\
    S_{TS}(N)&= C_{TS}\lambda_{TS}^N\\
    S_{ST}(N)&= C_{ST}\lambda_{ST}^N\\
    S_{T\perp}(N)&= C_{T\perp}\lambda_{TS}^N.
    \label{eq:fitSubspace}
\end{aligned}
\end{equation}
Note that from our general form Eq. \ref{eq:fitForm} we would expect that $S_0(N)$ is the sum of two exponential terms, with each $\lambda_{0,j}$ corresponding to an eigenvalue of $\hat{\Lambda}_G$ restricted to $\mathcal{H}_0$. However, we know that for trace-preserving noise $\langle\langle\mathbbm{1}|\hat{\Lambda}_G=\langle\langle\mathbbm{1}|$, which implies that one of the eigenvalues is $1$. 

We define two different subgroups $\overline{G}_1,\overline{G}_2\subseteq G$ for our benchmarking procedure. We will use $\overline{G}_1$ to construct $S_0(N)$ and $S_{T\perp}(N)$, and $\overline{G}_2$ to construct $S_{TS}(N)$ and $S_{ST}(N)$. We define 
\begin{align*}
\overline{G}_1&:=\{X^aZ^b\oplus \omega^\eta : a,b,\eta = 0,1,2\}\\ \overline{G}_2&:=\{Z^b\oplus \omega^\eta : b,\eta = 0,1,2\}. 
\end{align*}

For $\overline{G}_1$, we can define the following character functions and their corresponding projectors:
\begin{align*}
    \chi_{\overline 0}(X^aZ^b\oplus\omega^\eta) &=1\\
    \hat{P}_{\overline 0}&= \frac{1}{3}|\mathbbm{1}_T\rangle\rangle\langle\langle\mathbbm{1}_T|+|\mathbbm{1}_S\rangle\rangle\langle\langle\mathbbm{1}_S|\\
    \chi_{\overline{T\perp}}(X^aZ^b\oplus\omega^\eta) &=\omega^{-a}\\
    \hat{P}_{\overline {T\perp}}&= \frac{1}{3}|Z\rangle\rangle\langle\langle Z|
\end{align*}
We see that $\hat{P}_{\overline 0}$ projects into $2\overline{\mathcal{H}}_{\overline 0}\subseteq 2\mathcal{H}_0$ and $\hat{P}_{\overline{T\perp}}$ projects into $\overline{\mathcal{H}}_{\overline{T\perp}}\subseteq \mathcal{H}_{T\perp}$, as required. We also see that $\text{dim}(\overline{\mathcal{H}}_{\overline{T\perp}})=1$, so that $S_{T\perp}(N)$ will have the best possible relative error (see Section \ref{section:scaling}).

For $\overline{G}_2$, we can define the character functions and corresponding projectors
\begin{align*}
    \chi_{\overline{TS}}(Z^b\oplus\omega^\eta) &=\omega^{b-\eta}\\
    \hat{P}_{\overline {TS}}&=|T\rangle|S\rangle\langle T|\langle S|\\
    \chi_{\overline{ST}}(Z^b\oplus\omega^\eta) &=\omega^{-b+\eta}\\
    \hat{P}_{\overline {ST}}&=|S\rangle|T\rangle\langle S|\langle T|
\end{align*}
where $|T\rangle:=(|01\rangle+|10\rangle)/\sqrt{2}$ is the triplet state satisfying $Z|T\rangle=\omega|T\rangle$ and $|S\rangle:=(|01\rangle-|10\rangle)/\sqrt{2}$ is the singlet state. We again see that $P_{\overline{TS}}$ projects into $\overline{\mathcal{H}}_{\overline{TS}}\subseteq\mathcal{H}_{TS}$ and  $\text{dim}(\overline{\mathcal{H}}_{\overline{TS}})=1$, so that $S_{TS}(N)$ will also have the best possible relative error.

As our initial states, we choose
$$
|\rho_i\rangle\rangle = \left\{\begin{array}{ll}
|00\rangle\rangle, &i=0,T\! \! \perp \\
|01\rangle\rangle, &i=TS,ST\end{array}\right.
$$
Here, we've restricted ourselves to initial states that are a mixture of $Z$-basis product states, for ease of preparation.

As our measurement projectors, we choose
$$
|M_i\rangle\rangle = \left\{\begin{array}{ll}
|00\rangle\rangle+|11\rangle\rangle, &i=0,T\! 
\! \perp\\
|01\rangle\rangle, &i=TS,ST\\\end{array}\right.
$$
Here, we've restricted our measurement projectors to correspond to $Z$ measurements, for ease of measuring.

With these choices, the $S_i(N)$ are approximately
\begin{align*}
    S_i(N) & \approx \frac{\langle\langle M_i|\hat{P}_{\overline i}|\rho_i\rangle\rangle}{\text{dim}(\mathcal{\overline H}_{\overline i})}=\left\{\begin{array}{ll}
\frac{2}{3} , &{i}={0}\\
\frac{e^{-i\pi/3}}{3} , &{i}=T\! 
\! \perp\\
\frac{1}{4} , &{i}=TS,ST\\\end{array}\right.
\end{align*}

Note that $\lambda_{ST}=\lambda_{TS}^*$, so it is unnecessary to compute both $S_{TS}(N)$ and $S_{ST}(N)$. Note also that $\lambda_0$ and $\lambda_{T\perp}$ are both necessarily real, as are $C_0$ and $B$. The remaining parameters are complex. For convenience, we will rotate $S_{T\perp}(N)$ by $e^{i\pi/3}$ so that $S_{T\perp}(N)$ is approximately real.

\begin{figure}
    \centering
    \includegraphics[width=\columnwidth]{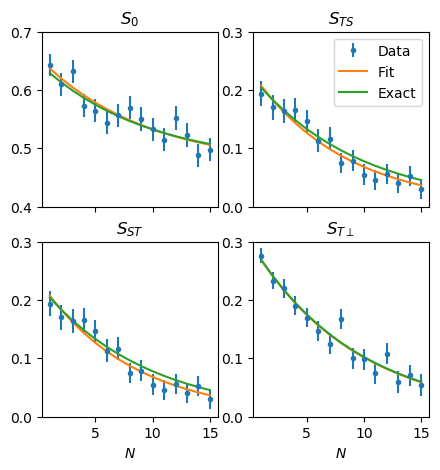}
    \caption{The predicted and measured character-weighted survival probability for a random error channel. The exact decay (green) is an exponential decay given by Eq. \ref{eq:fitSubspace}. We estimate $S_i(N)$ by applying random gates and measuring the final state (blue points). The data is fit to an appropriate function (orange) from which we estimate the fidelity.}. 
    \label{fig:IndividualDecays}
\end{figure}
\begin{figure}
    \centering
    \includegraphics[width=\columnwidth]{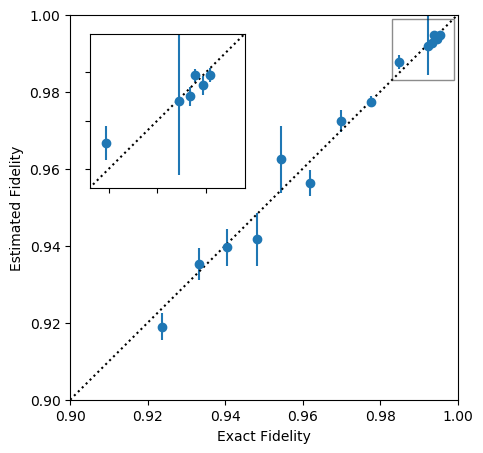}
    \caption{The exact and estimated fidelity for a selection of randomly generated error channels. Each estimate was based on data taken over 15 different lengths $N$. Each estimate was arrived at by applying a total of 150,000 benchmarking group elements. This is the same number of elements applied in the experiment described in \cite{baldwin2020subspace}. The diagonal line denotes the points where the exact and estimated fidelities are equal. The data agree with the line with a reduced $\chi^2$ value of $.9$, indicating good agreement. Note that the error bars are derived from statistical uncertainty in the data, and vanish in the limit of an infinite number of data points}
    \label{fig:FidelityEstimator}
\end{figure}

\begin{figure*}
    \centering
    \includegraphics[width=\textwidth]{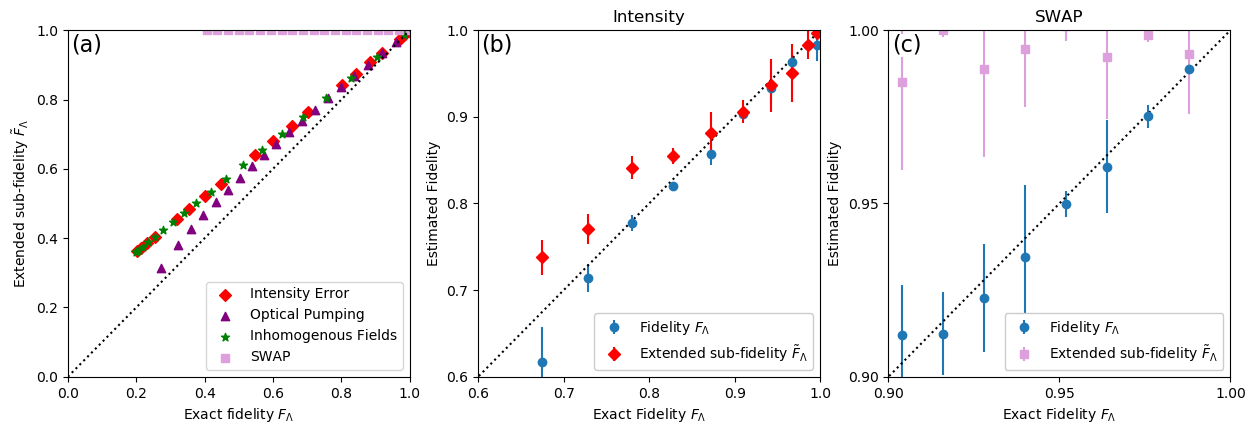}
    \caption{(a) The extended sub-fidelity $\tilde{F}_\Lambda$ of \cite{baldwin2020subspace} versus the exact fidelity $F_\Lambda$ that we can estimate in our paper, for a selection of error channels of varying strengths: intensity errors, which correspond to an overrotation $e^{-i\epsilon ZZ}$; optical pumping errors, which cause amplitude-damping on each qubit; inhomogeneous fields, which cause phase-damping on each qubit; and SWAP errors, which interchange the qubits. This plot corresponds to the exact value of both $F_\Lambda$ and $\tilde{F}_\Lambda$ that one estimates in an experiment. Note that while the $\tilde{F}_\Lambda$ agrees with $F_\Lambda$ in the limit $F_\Lambda\rightarrow 1$, in general the two do not agree, and there exists worst-case errors such as SWAP that $\tilde{F}_\Lambda$ cannot detect. (b,c) Simulation of an experiment that estimates $F_{\Lambda}$ versus $\tilde{F}_\Lambda$ for a total of $300,000$ unitaries, in the case of (b) intensity and (c) SWAP errors of varying strengths. These plots correspond to experiments that estimate the exact values shown in (a). We see that the difference between $F_\Lambda$ and $\tilde{F}_\Lambda$ can be discerned in a realistic experiment.}
    \label{fig:subspaceFidelity}
\end{figure*}
We demonstrate our method by generating random error channels and simulating our RB procedure. To generate a random error channel $\Lambda$ on a $d$-dimensional Hilbert space, we generate a random unitary on a $(d^2+d)$ dimensional Hilbert space and trace out $d^2$ auxiliary degrees of freedom; to adjust the fidelity, we take a convex combination of the resulting channel with the identity channel. All channels generated by this method are guaranteed to be completely positive trace-preserving (CPTP), thus valid error channels, and every CPTP channel can be generated via this method \cite{wood2015tensor}. For each error channel, we take data at 15 different values of $N$, and sample unitary operators at each value of $N$ until we have applied a total of $150,000$ unitary operators in total. For each string of unitary operators, we perform full state-vector simulation to apply the RB sequence of operators, and then generate a measurement outcome of $0$ or $1$ using the appropriate probability, and compute the character-weighted average. In Fig. \ref{fig:IndividualDecays}, we show the exact value of $S_i(N)$, the data we take to estimate $S_i(N)$, and the fit to $S_i(N)$ according to Eq. \ref{eq:fitSubspace} for a single random error channel $\Lambda$. 

From the fit data, we can estimate $F_\Lambda$ by applying Eq. \ref{eq:fidelityEst}:
\begin{equation}
    F_\Lambda =\frac{1+\lambda_{0}+8\lambda_{T\perp}+3\lambda_{TS}+3\lambda_{ST}+4}{20}.
\end{equation}
Note that the imaginary parts of $\lambda_{TS}$ and $\lambda_{ST}$ always cancel to give a real $F_\Lambda$ as expected. We use this formula to estimate the fidelity of our randomly generated error channels, and compare our estimate to the true fidelity in Fig. \ref{fig:FidelityEstimator}. We see that the true fidelity and the estimated fidelity agree within the error bars set by the uncertainty of our fits.

We can directly compare this with the original subspace RB method \cite{baldwin2020subspace}. That method served to estimate only $\lambda_{0}$ and $\lambda_{T\perp}$ ($t$ and $r$ in their notation), and they could only form a measure of gate fidelity using these quantities. They defined a so-called ``extended sub-fidelity" $\tilde{F}_\Lambda$, which they obtained by replacing $\lambda_{ST}$ and $\lambda_{TS}$ with the weighted average of the other eigenvalues: \rev{$\lambda_{ST}+\lambda_{TS} \approx 2\frac{1+\lambda_{0}+8\lambda_{T\perp}}{10}$. Explicitly, the extended sub-fidelity is given by \footnote{Our formula differs slightly from the corresponding formula in \cite{baldwin2020subspace}. Ref. \cite{baldwin2020subspace} considered approximating the process (also called entanglement) fidelity rather than the average fidelity; however, the average fidelity can be determined from the process fidelity\cite{horodecki1999general,nielsen2002simple}. To be consistent with the rest of our paper, we have translated their approximation of the process fidelity into the corresponding approximation of the average fidelity.} $$\tilde{F}_\lambda = \frac{16\lambda_{T\perp}+2\lambda_0+7}{25}.$$}\noindent
It is obvious that if $F_\Lambda\rightarrow 1$, $\tilde{F}_\Lambda\rightarrow 1$ as well, but the reverse is not necessarily true. We can compare the \rev{extended sub-fidelity} to the exact fidelity for the various noise sources explored \rev{in} \cite{baldwin2020subspace}. We consider intensity errors, which correspond to an overrotation $e^{-i\epsilon ZZ}$; optical pumping errors, which cause amplitude-damping on each qubit; inhomogenous fields, which cause phase-damping on each qubit; and SWAP errors, which interchange the qubits. The results are shown in Fig. \ref{fig:subspaceFidelity}. We see that while for most error sources $F_\Lambda\approx \tilde{F}_\Lambda$, there exist worse-case errors, such as SWAP, that cannot be detected by $\tilde{F}_\Lambda$. This was also noted in \cite{baldwin2020subspace} as a limitation of their method.

Our work also improves upon the original work in the mathematical assumptions needed to derive the benchmarking decays. Ref. \cite{baldwin2020subspace} derived their decay formulas under the assumption that their benchmarking set was of the form $\{U_T\oplus \sigma\phi_{U_T}:U_T\in G_T,\sigma=\pm\}$, where $\phi_{U_T}$ is some uncontrolled phase that occurs on the singlet space and $\sigma$ is a controllable phase between the singlet and triplet spaces. However, in practice they could not control $\sigma$ using a constant number of $U_{ZZ}$ gates. Instead, they implemented only $\{U_T\oplus \phi_{U_T}:U_T\in G_T\}$ and assumed the form of the decay would not change. In our work, by contrast, we have rigorously derived decay formulas for a group of gates that can be directly compiled into elementary symmetric gates using a constant number of $U_{ZZ}$.

We note that our method does require one additional capability that was not required in the original work: in order to estimate $S_{TS}(N)$, it is necessary to initialize and measure the $|01\rangle$ state. This requires additional experimental overhead to individually address and measure each qubit at the beginning and end of the benchmarking procedure. However, such overhead only contributes to the SPAM errors $\Lambda_P,\Lambda_M$, and does not affect our estimates of the entangling error. In any case, our method to measure $\lambda_0$ and $\lambda_{T\perp}$ does not require individual addressing, and can be viewed as a mathematically rigorous method to extract these parameters with no additional experimental requirements.

\section{Application: Leakage randomized benchmarking}
\label{section:leakage}
We may also use our generalized character RB to improve the leakage RB introduced in \cite{wood2018quantification}. In leakage RB, like subspace RB, one is given a group $G$ that preserves the splitting of the Hilbert space into subspaces $\mathcal H = \mathcal{H}_1\oplus\mathcal{H}_2$. In leakage RB, however, $\mathcal{H}_1\oplus\mathcal{H}_2$ does not represent the computational Hilbert space, and the goal is not to compute the average fidelity of the group operations. Instead, $\mathcal{H}_1$ represents the computational space of a quantum system (e.g. the two lowest-level states that encode a qubit), while $\mathcal{H}_2$ represents the leakage space outside the computational space. Leakage RB determines the average probability of ``leaking" from $\mathcal{H}_1$ to $\mathcal{H}_2$ or ``seeping" from $\mathcal{H}_2$ to $\mathcal{H}_1$. Noting that the probability of a state $|\rho\rangle\rangle$ being in subspace $\alpha=1,2$ is given by $\langle\langle\mathbbm{1}_\alpha|\rho\rangle\rangle$, define the leakage $L$ and seepage $S$ by
\begin{align}
    L := \int d\psi_1 \langle\langle\mathbbm{1}_2|\hat{\Lambda}|\psi_1\rangle\rangle = \frac{1}{d_1}\langle\langle\mathbbm{1}_2|\hat{\Lambda}|\mathbbm{1}_1\rangle\rangle\label{eq:leakage}\\
    S := \int d\psi_2 \langle\langle\mathbbm{1}_1|\hat{\Lambda}|\psi_2\rangle\rangle = \frac{1}{d_2}\langle\langle\mathbbm{1}_1|\hat{\Lambda}|\mathbbm{1}_2\rangle\rangle.
    \label{eq:seepage}
\end{align}
In addition, leakage RB determines the average fidelity restricted to the subspace $\mathcal{H}_1$
\begin{equation}
F_{\Lambda,1} = \int d\psi_1 \langle\langle\psi_1|\hat{\Lambda}|\psi_1\rangle\rangle.
\label{eq:restrictedFidelity}
\end{equation}
which is the appropriate measure of gate quality, since all computations take place in $\mathcal{H}_1$. Leakage RB is relevant for any system in which qubits are encoded in the subspace of a larger Hilbert space, which includes superconducting qubits \cite{gambetta2017building,chen2016measuring}, quantum dots, \cite{divincenzo2000universal,petta2005coherent,hanson2007universal,levy2002universal,andrews2019quantifying}, and trapped ions \cite{haffner2008quantum,hayes2020eliminating,stricker2020experimental}.

The original leakage RB could only be applied to a group 
\rev{\begin{equation}G = \{U_{1,b_1}\oplus \sigma U_{2,b_2}:(b_1,b_2)\in B_1\times B_2,\ \sigma=\pm 1\}\label{eq:originalLeakage}\end{equation}}\noindent
such that \rev{$G_1=\{U_{1,b_1}:b_1\in B_1\}$} and \rev{$G_2=\{U_{2,b_2}:b_2\in B_2\}$} form $2$-designs on their respective subspaces \footnote{Ref. \cite{wood2018quantification} originally claimed it was sufficient for \rev{$G_2$} to be a unitary 1-design, but this appears to be an error}. This is a very stringent condition, as it requires being able to independently control the computational and leakage subspaces. In many experimental implementations such control is not realistic; an experimental implementation of a gate \rev{$U_{1,b}$} on the computational subspace will naturally implement some \rev{$U_{2,b}$} on the leakage subspace. It is therefore desirable to develop a leakage RB that can be applied to more general groups.

Using our method, we can derive a leakage RB procedure that is more general than the one described in \cite{wood2018quantification}. Let $G$ be a group of unitary gates that preserve the subspaces of $\mathcal{H}$, and let $\Lambda$ be their shared error channel. To estimate $L$ and $S$, we will require that the only trivial representations of $G$ are $|\mathbbm{1}_1\rangle\rangle$ and $|\mathbbm{1}_2\rangle\rangle$, while to estimate $F_{\Lambda,1}$ we additionally require that the subrepresentation $\mathcal{H}_{1\perp}\subseteq \mathcal{H}_1\otimes\mathcal{H}_1$ orthogonal to $|\mathbbm{1}_1\rangle\rangle$ is an irrep of multiplicity $1$.

If we write our group $G$ as 
\rev{\begin{align*}
    G&=\{U_{b,\sigma}:b\in B,\ \sigma=\pm 1\}\\&=\{U_{1,b}\oplus \sigma U_{2,b}:b\in B,\ \sigma=\pm 1\}.
\end{align*}}\noindent
then the first condition is satisfied provided \rev{$\{U_{1,b}:b\in B\}$} and \rev{$\{U_{2,b}:b\in B\}$} are unitary $1$-designs, while the second condition is satisfied if provided these groups are unitary $2$-designs with dimensions $d_1\neq d_2$ (see Appendix \ref{appendix:choosingLeakageGroup} for proofs). Note that our requirements are significantly weaker than the original leakage RB, as we are only assuming the ability to implement an independent phase on the leakage space.

We outline our procedure for determining $L$, $S$, and $F_{\Lambda,1}$ for such groups $G$. Our procedure, like the original leakage RB, requires that SPAM errors do not mix the the subspaces $\mathcal{H}_1$ and $\mathcal{H}_2$, or at least that such mixing is negligible compared to the gate errors. In our derivations we will assume $\hat\Lambda_M=\hat\Lambda_P=\hat{\mathbbm{1}}$, although the generalization to errors that act only within the subspaces is trivial.

Our modified leakage RB procedure consists of the following steps:

\begin{enumerate}
    \item Choose an initial state $|\rho\rangle\rangle\in\mathcal{H}_1$ and measurement projector \rev{$|M\rangle\rangle = |\mathbbm{1}_1\rangle$}.
    \item For a given $N$, choose unitaries \rev{$U_0,U_1,...,U_N\in G$} randomly and uniformly. Compute $U_{N+1}=U_1^\dagger\cdots U_N^\dagger$.
    \item Prepare the state $|\rho\rangle\rangle$. Apply the gates $(U_1U_0), U_2,..., U_{N+1}$ sequentially, where $({U}_1{U}_0)$ is compiled as a single element of $G$.
    \item Perform a measurement of the observable $M$ to determine if the state is still in $\mathcal{H}_1$.
    \item Repeat steps 2-4 many times, to estimate the \rev{trivial} character-weighted survival probability
    \rev{\begin{equation}
        S_0(N)=\frac{1}{|G|^{N+1}}\sum_{U_0,...,U_N\in G} \text{Pr}_{U_0,...,U_N}
        \label{eq:modifiedSurvival}
    \end{equation}}\noindent
    where $\text{Pr}_{U_0,...,U_{N+1}}$ is the probability of remaining in $\mathcal{H}_1$ after applying gates \rev{$(U_1U_0),...,U_{N+1}$} to $|\rho\rangle\rangle$.
    \item Repeat steps 2-5 for different values of $N$.
    \item Fit the survival probability to a function of the form
    \begin{equation}
        S_0(N) = A\lambda^N+B
        \label{eq:modifiedFitForm}
    \end{equation}
    where $A$, $B$, and $\lambda$ are independent of $N$.
    \item Estimate $L$ and $S$ as
    \begin{align}
    L & = (1-B)(1-\lambda)\label{eq:leakageEst}\\
    S & = B(1-\lambda)\label{eq:seepageEst}
    \end{align}
    \item Use the original character RB (section \ref{section:generalizedRB}) to measure the character-weighted survival probability $S_{1\perp}$ associated to the irrep $\mathcal{H}_{1\perp}$. Fit 
    $$S_{1\perp}(N)=C\lambda_{1\perp}^N$$
    to estimate $\lambda_{1\perp}$.
    \item Estimate $F_{\Lambda,1}$ as 
    \begin{equation}
        F_{\Lambda,1} = \frac{(d_1^2-1)\lambda_{1\perp}+(d_1+1)(1-L)}{d_1^2+d_1}.
        \label{eq:restrictedFidelityEst}
    \end{equation}
\end{enumerate}

In the remainder of this section, we prove the correctness of this procedure and provide an example of such leakage RB.

\subsection{Deriving $L$ and $S$}
\label{section:LSProof}
Written out explicitly, the zeroth character-weighed survival probability is
\begin{equation*}
    S_0(N) = \langle\langle\mathbbm{1}_1|\hat{\Lambda}\hat{\Lambda}_G^N\hat{P}_0|\rho\rangle\rangle.
\end{equation*}
where $\hat{P}_0$ is the projector onto the trivial irrep, and we have made the same substitutions as in Section \ref{section:decayDerivation} to reduce the sum over $\{U_0,...,U_N\}$ to $G$-twirls and a projector. We know from Thm. \ref{thm:twirlForm} that $\hat{\Lambda}_G$ has a block-diagonal form $\hat{\Lambda}_G=\bigoplus_{i}\hat{Q}_i\otimes \hat{\mathbbm{1}}_i$, where $i$ indexes the irreps. Because $\hat\Lambda_G$ is multiplied by the projector $\hat{P}_0$ in Eq. \ref{eq:modifiedSurvival}, we may ignore all terms except $\hat{Q}_0\otimes \mathbbm{1}_0$. In terms of the eigendecomposition of $\hat{Q}_0$, we may write $\hat{Q}_0\otimes\mathbbm{1}_0=|e_0\rangle\rangle\langle\langle\overline{e}_0|+\lambda|e_1\rangle\rangle\langle\langle\overline{e}_1|$, so that
$$
{S}_0(N) = \langle\langle\mathbbm{1}_1|\hat{\Lambda}|e_0\rangle\rangle\langle\langle\overline{e}_0|\rho\rangle\rangle+\langle\langle\mathbbm{1}_1|\hat{\Lambda}|e_1\rangle\rangle\langle\langle\overline{e}_1|\rho\rangle\rangle\lambda^N
$$
where we have used the fact, noted in Section \ref{section:subspace}, that one eigenvalue of $\hat{Q}_0$ is always $1$. This justifies the fit Eq. \ref{eq:modifiedFitForm}.

So far, we have simply repeated the steps in Section \ref{section:decayDerivation} with slight modifications. However, in order to estimate $L$ and $S$ we will need to explicitly determine the eigendecomposition of \rev{$\hat{Q}_0\otimes\mathbbm{1}_0$}. We first note that the $\hat{P}_0$ subspace is spanned by the orthonormal vectors
$$
\frac{1}{\sqrt{d_1}}|\mathbbm{1}_1\rangle\rangle:= |\hat{\mathbbm{1}}_1\rangle\rangle\qquad\frac{1}{\sqrt{d_2}}|\mathbbm{1}_2\rangle\rangle:= |\hat{\mathbbm{1}}_2\rangle\rangle.
$$
Thus in terms of these basis vectors, we may write
$$
\hat{Q}_0\otimes \mathbbm{1}_0 = |\hat{\mathbbm{1}}_\alpha\rangle\rangle Q_{\alpha\beta}\langle\langle\hat{\mathbbm{1}}_\beta|
$$
for some constants $Q_{\alpha\beta}$. Noting that $M_{\alpha\beta}=\langle\langle\hat{\mathbbm{1}}_\alpha|\hat{\Lambda}_G|\hat{\mathbbm{1}}_\beta\rangle\rangle=\langle\langle\hat{\mathbbm{1}}_\alpha|\hat{\Lambda}|\hat{\mathbbm{1}}_\beta\rangle\rangle$, we can use the definitions of $L$ and $S$, (Eqs. \ref{eq:leakage} and \ref{eq:seepage}) to determine the constants $Q_{\alpha\beta}$:
$$
Q_{\alpha\beta}=\left(\begin{matrix} 1-L & \sqrt{\frac{d_2}{d_1}}S \\\sqrt{\frac{d_1}{d_2}}L & 1-S\end{matrix}\right)_{\alpha\beta}.
$$
From the explicit form of $Q_{\alpha\beta}$, we can determine the eigendecomposition of $\hat{Q}_0\otimes\mathbbm{1}_0$ via straightforward algebra \cite{wood2018quantification,baldwin2020subspace}:
\begin{align*}
|e_0\rangle\rangle &= \frac{S}{\sqrt{d_1}(L+S)}|\hat{\mathbbm{1}}_1\rangle\rangle+\frac{L}{\sqrt{d_2}(L+S)}|\hat{\mathbbm{1}}_2\rangle\rangle\\
|\overline{e}_0\rangle\rangle &= \sqrt{d_1}|\hat{\mathbbm{1}}_1\rangle\rangle + \sqrt{d_2}|\hat{\mathbbm{1}}_2\rangle\rangle\\
|e_1\rangle\rangle &= \sqrt{d_2}|\hat{\mathbbm{1}}_1\rangle\rangle - \sqrt{d_1}|\hat{\mathbbm{1}}_2\rangle\rangle\\
|\overline{e}_1\rangle\rangle &= \frac{L}{\sqrt{d_2}(L+S)}|\hat{\mathbbm{1}}_1\rangle\rangle-\frac{S}{\sqrt{d_1}(L+S)}|\hat{\mathbbm{1}}_2\rangle\rangle\\
\lambda & = 1-L-S
\end{align*}

Putting this together, we can evaluate the zeroth character-weighted survival probability as
\begin{equation*}
{S}_0(N) = \frac{S}{L+S}+
\frac{L}{L+S}(1-L-S)^{N+1}
\end{equation*}
We then have that $B = \frac{S}{L+S}$, which can be combined with $\lambda=1-L-S$ to immediately give Eqs. \ref{eq:leakageEst} and \ref{eq:seepageEst}.

\subsection{Deriving $F_{\Lambda,1}$}
\label{section:restrictedFidelityProof}
To establish Eq. \ref{eq:restrictedFidelityEst}, we first prove the following:
\begin{equation}
F_{\Lambda,1} = \frac{\Tr(\hat\Lambda\hat{P}_{11})+d_1(1-L)}{d_1^2+d_1}
\label{eq:restrictedTraceFidelity}
\end{equation}
where $\hat{P}_{11}$ is the projector onto $\mathcal{H}_1\otimes\mathcal{H}_1$. We use a similar method as in our proof of Eq. \ref{eq:TraceFidelity}. We first note that the restricted average fidelities of $\hat{\Lambda}$ and $\hat{P}_{11}\hat{\Lambda}\hat{P}_{11}:=\hat{\Lambda}_{11}$ are equal. $\hat{\Lambda}_{11}$ is an error channel restricted to the $\mathcal{H}_1$ subspace. We can twirl $\hat{\Lambda}_{11}$ by the full unitary group on $\mathcal{H}_1$ to get a depolarizing channel
$$
(\Lambda_{11})_G = p\mathbbm{1}_1+q\frac{1}{d_1}|\mathbbm{1}_1\rangle\rangle\langle\langle\mathbbm{1}_1|.
$$
Note that we have $p$ and $q$ rather than $p$ and $(1-p)$ as in Eq. \ref{eq:twirlDepolarizing}; this is because $\hat{\Lambda}_{11}$ is not necessarily trace-preserving. We can directly compute $F_{(\Lambda_{11})_G}=p+\frac{q}{d_1}$. Similarly, we can also directly compute $\Tr\left((\hat{\Lambda}_{11})_G\right)=pd_1^2+q$. Finally, we can directly compute $p+q=\frac{1}{d_1}\langle\langle \mathbbm{1}_1|(\hat{\Lambda}_{11})_G|\mathbbm{1}_1\rangle\rangle=\frac{1}{d_1}\langle\langle \mathbbm{1}_1|\hat{\Lambda}|\mathbbm{1}_1\rangle\rangle=1-L$. Combining these three equations gives Eq. \ref{eq:restrictedTraceFidelity}.

To estimate $\Tr(\hat\Lambda\hat{P}_{11})$, we can divide this trace up into two pieces:
$$
\Tr(\hat\Lambda\hat{P}_{11}) = \langle\langle\hat{\mathbbm{1}}_1|\hat{\Lambda}|\hat{\mathbbm{1}}_1\rangle\rangle +\Tr(\hat\Lambda\hat{P}_{1\perp})= (1-L) +\Tr(\hat\Lambda\hat{P}_{1\perp})
$$
where $\hat{P}_{1\perp}$ is the projector onto $\mathcal{H}_{1\perp}$. The latter trace is simply $(d_1^2-1)\lambda_{1\perp}$. Plugging this in to Eq. \ref{eq:restrictedTraceFidelity} gives Eq. \ref{eq:restrictedFidelityEst} as desired.

\subsection{Example: Two-qubit logical encodings}

\rev{Here, we illustrate the advantages of our leakage RB over the original leakage RB of \cite{wood2018quantification} via a single-qubit example where \cite{wood2018quantification} is not applicable.}. We consider an encoding of a single logical qubit into the $S_z=0$ subspace of two physical qubits. This encoding is frequently used in quantum dot qubits \cite{petta2005coherent,hanson2007universal,levy2002universal}. The computational space $\mathcal{H}_1$ is spanned by
$$
|0\rangle:=\frac{|01\rangle-|10\rangle}{\sqrt 2},\qquad |1\rangle:=\frac{|01\rangle+|10\rangle}{\sqrt 2}
$$
and the leakage space $\mathcal{H}_2$ is spanned by
$$
|2\rangle:=|00\rangle,\qquad |3\rangle:=|11\rangle.
$$

Let's assume we implement single-qubit rotations on our computational space by the operators
$$
R_X = X_C\oplus Z_L\qquad R_Z = Z_C\oplus \frac{X_L+Z_L}{\sqrt 2},
$$
where implementing an $X$ or $Z$ rotation on the computational space naturally induces a specific rotation on the leakage space.

\begin{figure}
    \centering
    \includegraphics[width=\columnwidth]{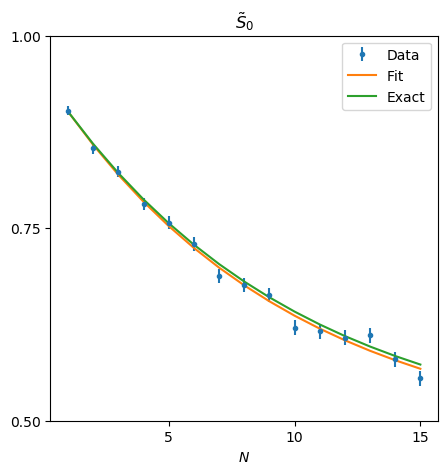}
    \caption{The predicted and measured $S_0(N)$ for a single randomly generated error channel. The actual decay (green) is an exponential decay given by Eq. \ref{eq:modifiedFitForm}. We estimate $S_0(N)$ by applying random gates and measuring the final state (blue points). The data is fit to a function of the form of Eq. \ref{eq:modifiedFitForm}, from which we estimate $L$ and $S$.}
    \label{fig:ExampleFitLeakage}
\end{figure}

\begin{figure}
    \centering
    \includegraphics[width=\columnwidth]{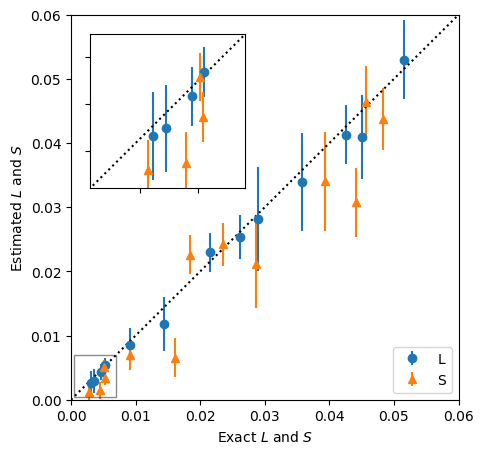}
    \caption{The exact and estimated leakage and seepage for a selection of randomly generated error channels. Each estimate was based on data taken over 15 different lengths $N$. Each estimate was arrived at by applying a total of 300,000 unitary group elements. The diagonal line denotes the points where the exact and estimated fidelities are equal. The data agree with this line with a reduced $\chi^2$ value of $1.3$, indicating good agreement.}
    \label{fig:LeakageEstimator}
\end{figure}

We will take our benchmarking group to be the group generated by these two rotations, $G=\langle R_X,R_Z \rangle$. This group has a total of 16 elements. It cannot be written as \rev{direct sum of a group acting on $\mathcal{H}_1$ and a group acting on $\mathcal{H}_2$ as in Eq. \ref{eq:originalLeakage}}, so the leakage RB \rev{of \cite{wood2018quantification}} does not apply. However, elementary calculation shows that the natural representation of this group contains exactly two trivial irreps, spanned by $|\mathbbm{1}_1\rangle\rangle$ and $|\mathbbm{1}_2\rangle\rangle$, and we can therefore use our procedure to estimate $L$ and $S$.

We illustrate this method by generating random error channels and simulating the RB procedure. In Figs. \ref{fig:ExampleFitLeakage}, we show the exact value of $S_0(N)$, the data we take to estimate $S_0(N)$, and the fit to $ S_0(N)$ according to Eq. \ref{eq:modifiedFitForm}. In Fig. \ref{fig:LeakageEstimator}, we repeat the same fitting procedure for a set of randomly generated error channels, and estimate $L$ and $S$ using Eq. \ref{eq:leakageEst}. We see that the true values of $L$ and $S$ and our estimate for $L$ and $S$ agree within the error bars set by the uncertainty in our fits.

We cannot apply our method to find $F_{\Lambda,1}$ because in this example $\mathcal{H}_{2\perp}$ and $\mathcal{H}_{1\perp}$ share an irrep. This reflects the overall difficulty in applying leakage RB to physically realistic circumstances. While this work provides the most widely applicable method for leakage RB currently available, more work is needed to develop a truly general procedure.

\section{Application: Matchgate RB}
\label{section:matchgate}
We can also use our method to introduce a new procedure for scalably benchmarking circuits made of matchgates. Matchgates are 2-qubit gates of the form
$$
G(A,B) = \left(\begin{matrix}a_{11}&0&0&a_{12}\\0&b_{11}&b_{12}&0\\0&b_{21}&b_{22}&0\\a_{21}&0&0&a_{22}\end{matrix}\right)
$$
with $\det(A)=\det(B)$. In other words, a matchgate acts as $A$ on the even parity subspace spanned by $\{|00\rangle,|11\rangle\}$ and as $B$ on the odd parity subspace spanned by $\{|01\rangle,|10\rangle\}$. Without loss of generality we may assume $\det(A)=\det(B)=1$. The set of matchgates acting on a line of nearest neighbors is efficiently simulable \cite{valiant2002quantum,terhal2002classical,jozsa2008matchgates,brod2016efficient}. However matchgates acting on next-nearest-neighbors \cite{jozsa2008matchgates} or acting on any nontrivial connectivity graph \cite{brod2012geometries,brod2014computational} are universal, as are matchgates plus arbitrary one-qubit gates \cite{imamog1999quantum,terhal2002classical}, matchgates plus a single $G(A,B)$ with $\det(A)\neq\det(B)$ \cite{brod2011extending}, matchgates acting on entangled input states \cite{hebenstreit2020computational}, and matchgates plus adaptive measurements \cite{hebenstreit2020computational}.  Implementations of arbitrary matchgates have been proposed for trapped atom systems \cite{herrera2014infrared} and have been experimentally demonstrated in photonic systems \cite{ramelow2010matchgate}.

We will derive a benchmarking procedure that determines the average fidelity of circuits composed of matchgates using a number of experiments that scales polynomially in the number of qubits. Our method is the matchgate equivalent of traditional Clifford RB, which characterizes the average fidelity of circuits composed of Hadamard, phase, and CNOT gates, and also requires a number of experiments that scales polynomially in the number of qubits. However, we will see that benchmarking matchgate circuits requires the full machinery of non-multiplicity-free character RB.

\subsection{The matchgate group}

Consider a line of $n$ qubits with nearest-neighbor connectivity. Let $G$ be the {\bf matchgate group} on $n$ qubits, the group of all unitaries generated from nearest-neighbor matchgates. Naively, $G$ could contain arbitrarily long circuits of matchgates. However, one can prove that every element of $G$ can be realized using circuits of at most $4n^3$ nearest-neighbor matchgates \cite[Thm.~5]{jozsa2008matchgates}. We will provide a simplified proof of this fact below.

Following \cite{terhal2002classical,jozsa2008matchgates}, our primary tool to understand $G$ will be the Jordan-Wigner transformation \cite{jordan1928uber}. Define $2n$ Majorana operators $\{c_i\}$ as
\begin{align*}
c_{2k-1}&=Z_1\cdots Z_{k-1}X_k\\
c_{2k}&=Z_1\cdots Z_{k-1}Y_k
\end{align*}
for $k=1,...,n$. The $\{c_m\}$ are Hermitian operators satisfying $\{c_\ell,c_m\}=2\delta_{\ell m}$. Polynomials in the $\{c_m\}$ form a Hermitian basis for the space of all density matrices, so a unitary $U$ is defined by its action on the $\{c_m\}$ up to a potential phase. Because of our restriction $\det(A)=\det(B)=1$, there is no phase freedom on the matchgates or any product of matchgates, so the action of $U\in G$ is entirely determined by its action on the $\{c_m\}$. We make two claims \cite{jozsa2008matchgates}:
\begin{claim}
    Every $U\in G$ in the matchgate group acts on the Majorana operators as a proper rotation. In other words, there exists some $R\in SO(2n)$ such that $Uc_\ell U^\dagger = R_{\ell m}c_m$.
    \label{claim:rotation}
\end{claim}
\begin{claim}
    Any unitary operator $U\in U(2^n)$ that acts on the Majorana operators as a proper rotation is in the matchgate group $G$. In particular, such a $U$ can be decomposed into a product of at most $2n^3$ nearest-neighbor matchgates.
    \label{claim:decomposition}
\end{claim}
These two claims together imply that the matchgate group is isomorphic to $SO(2N)$, and that every element of the matchgate group can be efficiently implemented in a quantum circuit. In particular, this shows that the matchgate group is a compact group, thus we can apply character RB.

\subsubsection{Proof of claims}

\begin{proof}[Proof of Claim \ref{claim:rotation}]
We provide a simplification of the proof in \cite{jozsa2008matchgates}. We prove that a nearest-neighbor matchgate acting on qubits $k$ and $k+1$ acts as a rotation mixing $c_{2k-1}$, $c_{2k}$, $c_{2k+1}$, and $c_{2k+2}$, and that all such rotations are realized by matchgates. It then follows that all products of matchgates also act as rotations on the Majorana operators.

Without loss of generality, we can restrict ourselves to $k=1$, so our Majorana operators are given by
\begin{align*}
    c_1 &= X_1 & c_3&=Z_1X_2\\
    c_2 &= Y_2 & c_4&=Z_1Y_2.
\end{align*}
We can write an infinitesimal matchgate as $U =\mathbbm{1}-i\epsilon M$, where $M$ must be of the form
$$
\alpha_{12}Z_1-\alpha_{13}Y_1X_2-\alpha_{14}Y_1Y_2+\alpha_{23}X_1X_2+\alpha_{24}X_1Y_2 +\alpha_{34}Z_2
$$
with $\alpha_{ab}\in \mathbbm{R}$. One can directly check that $U$ satisfies
\begin{align*}
Uc_1U^\dagger &= c_1+2\epsilon\alpha_{12}c_2+2\epsilon\alpha_{13}c_3+2\epsilon\alpha_{14} c_4\\
Uc_2U^\dagger &= -2\epsilon \alpha_{12}c_1+c_2+2\epsilon\alpha_{23}c_3+2\epsilon\alpha_{24} c_4\\
Uc_3U^\dagger &= -2\epsilon \alpha_{13}c_1-2\epsilon\alpha_{23}c_2+c_3+2\epsilon\alpha_{34} c_4\\
Uc_4U^\dagger &= -2\epsilon \alpha_{14}c_1-2\epsilon\alpha_{24}c_2-2\epsilon\alpha_{34}c_3+c_4
\end{align*}
so that $Uc_iU^\dagger=R_{ij}c_j$ with
$$
R = \mathbbm{1}+2\epsilon\left(\begin{matrix}
0 & \alpha_{12} & \alpha_{13} &\alpha_{14}\\
-\alpha_{12} & 0 & \alpha_{23} &\alpha_{24}\\
-\alpha_{13}&-\alpha_{23} & 0 &\alpha_{34}\\
-\alpha_{14}&-\alpha_{24} & -\alpha_{34}&0
\end{matrix}\right)
$$
We therefore see that infinitesimal matchgates generate the whole Lie algebra $\mathfrak{so}(4)$ of real antisymmetric matrices. By exponentiating the infinitesimal matchgates, we generate the full set of matchgates; in this process, we generate the full group $SO(4)$ as well.
\end{proof}

\begin{proof}[Proof of Claim \ref{claim:decomposition}]
We note, following \cite{jozsa2008matchgates}, that every $R\in SO(2n)$ can be decomposed into $n(2n-1)$ rotations that act as the identity on all but $2$ basis elements $c_\ell,c_m$ by the Hoffman algorithm \cite{raffenetti1969parametrization,hoffman1972generalization}. In turn, a rotation mixing $c_\ell$ and $c_m$ with $\ell<m$ can be decomposed into a product of $s:=\left(\lceil \frac{m}{2}\rceil-\lceil\frac{\ell}{2}\rceil-1\right)$ rotations that exchange $(c_\ell\leftrightarrow c_{\ell+2})$, $(c_{\ell+2}\leftrightarrow c_{\ell+4})$, ..., $(c_{\ell+2s-2}\leftrightarrow c_{\ell+2s})$, followed by a rotation that mixes $c_{\ell+2s}$ and $c_m$, followed by $s$ rotations that exchange $(c_{\ell+2s}\leftrightarrow c_{\ell+2s-2})$, $(c_{\ell+2s-2}\leftrightarrow c_{\ell+2n-4})$, ...,  $(c_{\ell+2}\leftrightarrow c_{\ell})$. Each of these rotations only involve Majorana operators associated to neighboring qubits, and thus can be written as a matchgate. Thus, $R$ can be realized as the product of a total of $n(2n-1)(2s+1)<4n^3$ matchgates, as claimed.
\end{proof}

We note that an arbitrary rotation between two Majorana operators 
$$
\left(\begin{matrix}c_\ell\\c_m\end{matrix}\right)\rightarrow \left(\begin{matrix}\cos(\theta) &\sin(\theta)\\-\sin(\theta)&\cos(\theta)\end{matrix}\right)\left(\begin{matrix}c_\ell\\c_m\end{matrix}\right)
$$
is generated by the unitary $U=e^{\frac{\theta}{2}c_\ell c_m}$. In the case where $\left|\lceil\frac{m}{2}\rceil-\lceil\frac{\ell}{2}\rceil\right|\leq 1$, this $U$ is a nearest-neighbor matchgate. For example, if $\ell=3$, $m=5$, then we have $U=e^{-i\frac{\theta}{2}Y_2X_3}$. Thus, the above decomposition of $R$ into $<4n^3$ two-Majorana rotations gives an explicit formula for the matchgates needed to construct $R$. We provide Python code to realize the Hoffman decomposition of $R$ into elementary rotations, as well as the reduction of $R$ to a matchgate circuit, at \cite{claes2020orthogonal}.

\subsection{Irreps of the matchgate group}

We want to understand how the natural representation of $G$ decomposes into irreps. This is most convenient in the basis of polynomials of $\{c_m\}$. Note that $c_m^2=1$, so our polynomials are at most degree $1$ in any given $c_m$ and there are $4^N$ such polynomials. Explicitly, an orthonormal basis of $\mathcal{H}\otimes\mathcal{H}$ is given by
\begin{align*}
    \frac{1}{2^{N/2}}\mathbbm{1} := & |\hat{\mathbbm{1}}\rangle\rangle\\
    \frac{1}{2^{N/2}}c_{m_1} :=& |m_1\rangle\rangle & 1\leq m_1 \leq 2n\\
    \frac{1}{2^{N/2}}c_{m_1}c_{m_2} :=& |m_1m_2\rangle\rangle& 1\leq m_1<m_2\leq 2n\\
    & \vdots &\vdots \qquad\\
    |m_1\cdots m_{2n-1}\rangle\rangle\span & 1\leq m_1<\cdots \leq 2n\\
    \span |1\cdots 2n\rangle\rangle.&
\end{align*}
Define $\mathcal{H}_i:=\text{span}\{|m_1\cdots m_i\rangle\rangle\}$ to be the space spanned by degree-$i$ basis elements, for each $i=0,...,2n$. Then $\mathcal{H}_i\simeq \bigwedge^i\mathbbm{C}^{2n}$, the $i$-fold wedge product of $\mathbbm{C}^{2n}$. It's clear that $\hat{U}$ preserves each $\mathcal{H}_i$, so that each $\mathcal{H}_i$ is a subrepresentation. On $\mathcal{H}_1$, $\hat{U}$ acts as the rotation operator $R$ associated to $U$:
$$
\hat{U}|i_1\rangle\rangle = R_{i_1j_1}|j_1\rangle\rangle.
$$
On general $\mathcal{H}_i$, $\hat{U}$ acts as the wedge product of the rotation operator:
\begin{align*}
\hat{U}|\ell_1\cdots \ell_i\rangle\rangle = \sum_{\mathclap{\substack{m_1<\cdots <m_i\\ \sigma\in S^i}}} (-1)^\sigma R_{\ell_1m_{\sigma 1}}\cdots R_{\ell_i m_{\sigma i}}|m_1\cdots m_i\rangle\rangle.
\end{align*}

\begin{claim}
The natural representation of the matchgate group decomposes into the irreps
$$\mathcal{H}_0\oplus \mathcal{H}_1\oplus \cdots \oplus\mathcal{H}_{n,1}\oplus \mathcal{H}_{n,2}\oplus \cdots\oplus \mathcal{H}_{2n-1}\oplus \mathcal{H}_{2n}.$$
where $\mathcal{H}_n=\mathcal{H}_{n,1}\oplus\mathcal{H}_{n,2}$. Explicitly, we have
\begin{align*}
    \mathcal{H}_{n,1} & = \text{span}\{|\ell_1\cdots\ell_n\rangle\rangle+i^n(-1)^{\sigma(\ell,m)}|m_1\cdots m_n\rangle\rangle\}\\
    \mathcal{H}_{n,2} & = \text{span}\{|\ell_1\cdots\ell_n\rangle\rangle-i^n(-1)^{\sigma(\ell,m)}|m_1\cdots m_n\rangle\rangle\}
\end{align*}
where $\{m_a\}$ is the complement of $\{\ell_a\}$ and $\sigma(\ell,m)$ is the permutation that takes $(\ell_1,...,\ell_n,m_1,...,m_{n})\mapsto (1,...,2n)$. Note that if $n$ is even these are real representations, while for $n$ odd these representations are complex conjugates of each other. The irreps $\mathcal{H}_i$ and $\mathcal{H}_{2n-i}$ are isomorphic for $i\neq n$, but no other irreps are isomorphic to each other. 
\label{claim:decompositionMatchgate}
\end{claim}
\begin{proof}
Define the Hodge star operator $*:\mathcal{H}_i\rightarrow\mathcal{H}_{2n-i}$ by
$$
*|\ell_1\cdots\ell_i\rangle\rangle = (-1)^{\sigma(\ell,m)}|m_1\cdots m_{2n-i}\rangle\rangle
$$
where $\{m_a\}$ is the complement of $\{\ell_a\}$ and $\sigma(\ell,m)$ is the permutation that takes $(\ell_1,...,\ell_i,m_1,...,m_{2n-i})\mapsto (1,...,2n)$. It is straightforward to show that $*$ commutes with the action of $U$, and thus provides the isomorphism of representations $\mathcal{H}_i\simeq \mathcal{H}_{2n-i}$ when $i\neq n$. We defer the proof that the $\mathcal{H}_i$, $\mathcal{H}_{n,1}$, and $\mathcal{H}_{n,2}$ are in fact irreducible to chapter 4 of \cite{knapp2001representation}.
\end{proof}

\subsection{Benchmarking the matchgate group}

Let $\overline{G}\subset G$ be the subgroup of the matchgate group generated by $R\in SO(2n)$ with $R$ diagonal. Such an $R$ is always of the form $R = \text{diag}\{\sigma_1,...,\sigma_{2n}\}$ with $\sigma_1\sigma_2\cdots\sigma_{2n}=1$. The action on a state $|m_1\cdots m_i\rangle\rangle\in\mathcal{H}_i$ is given by 
$$
\hat{U}|m_1\cdots m_i\rangle\rangle = \sigma_{i_1}\cdots \sigma_{i_m}|m_1\cdots m_i\rangle\rangle
$$
and therefore the states $|i_1\cdots i_m\rangle\rangle$ are the irreps of the natural representation of $\overline{G}$. Because of the constraint $\sigma_1\sigma_2\cdots\sigma_{2N}=1$, each irrep has multiplicity 2, with the irrep spanned by $|m_1\cdots m_i\rangle\rangle$ isomorphic to the irrep spanned by $|\ell_1\cdots \ell_{2n-i}\rangle\rangle$ with $\{\ell_a\}$ the complement of $\{m_a\}$. For each $i=0,...,n$, we can define a character function and corresponding projector
\begin{align*}
\chi_{\overline i}(R) =& \sigma_1\cdots\sigma_{i}\\
\hat{P}_{\overline i}=&|1\cdots i\rangle\rangle\langle\langle 1\cdots i|
\\&+ |(i+1)\cdots 2n\rangle\rangle\langle\langle (i+1)\cdots 2n|.    
\end{align*}
These projectors project into the multiplicty-two irreps $\mathcal{H}_i\oplus\mathcal{H}_{2n-i}$ for $\overline{i}=0,...,(n-1)$, and project into the two inequivalent irreps $\mathcal{H}_{n,1}\oplus\mathcal{H}_{n,2}$ for $\overline i=n$.

As our initial state, for each $i=0,...,n$ we choose
$$
|\rho_i\rangle\rangle = \left\{\begin{array}{ll}|0\cdots + \cdots 0 \rangle\rangle, & i=2k-1\\|0\cdots 0\rangle\rangle, & i=2k.\end{array}\right.
$$
where $k$th qubit is in the $+$ state of the $X$ operator for $i=2k-1$. Provided we can prepare both $X$-basis and $Z$-basis single qubit states, we can prepare $|\rho_i\rangle\rangle$.

As our measurement projector, for each $i=0,...,n$ we choose
$$
|M_i\rangle\rangle = \left\{\begin{array}{ll}\frac{1}{2}(X_{k}+\mathbbm{1}), & i=2k-1\\\frac{1}{2}\left(\prod_{\alpha>n-k}Z_\alpha+\mathbbm{1}\right), & i=2k.\end{array}\right.
$$
For $i=2k-1$, this corresponds to a measurement of the $k$th qubit in the $X$ basis, while for $i=2k$ this corresponds to a measurement of the product of the last $k$ qubits in the $Z$ basis.

With these choices, the $S_i(N)$ are approximately
$$
    S_i(N)  \approx \frac{\langle\langle M_i|\hat{P}_{\overline i}|\rho_i\rangle\rangle}{\text{dim}(\mathcal{\overline H}_{\overline i})}=\left\{\begin{array}{ll}1, & i= 0\\\frac{1}{2}, & 1\leq i\leq n\end{array}\right.
$$
and the relative uncertainty does not depend on the number of qubits. This is therefore a scalable method to benchmark the matchgate group.

The form of the decay is given by
\begin{equation}
    S_i(N)=\left\{\begin{array}{ll}C_0\lambda_0^N+B,&i=0\\
    C_{i,1}\lambda_{i,1}^N+C_{i,2}\lambda_{i,2}^N, &1\leq i\leq n.
    \end{array}\right.
    \label{eq:fitFormMatchgates}
\end{equation}
For each $i$, either $\lambda_{i,1},\lambda_{i,2},C_{i,1},C_{i,2}\in \mathbbm{R}$ or $\lambda_{i,1}=\lambda_{i,2}^*$ and $C_{i,1}=C_{i,2}^*$, since $S_i(N)$ is always real. For the case of $i=n$, we know that the former case holds when $n$ is even and the latter when $n$ is odd, by Claim \ref{claim:decompositionMatchgate}. For $1\leq i<n$, one should assume whichever case gives the best fit. Note that in all cases, we fit at most $4$ real parameters.

As an example, we simulate a noisy implementation of the matchgate group on $n=3$ qubits. In Fig. \ref{fig:ExampleFitMatchgates}, we show the exact value of $S_i(N)$, the data we take to estimate $S_i(N)$, and the fit to $S_i(N)$ according to Eq. \ref{eq:fitFormMatchgates} for a single random error channel $\Lambda$. In Fig. \ref{fig:FidelityEstimatesMatchgates}, we do the same fitting procedure for a set of randomly generated error channels, and estimate their fidelity. We see that the true fidelity and the estimated fidelity agree within the error bars set by the uncertainty of our fits.
\begin{figure}
    \centering
    \includegraphics[width=\columnwidth]{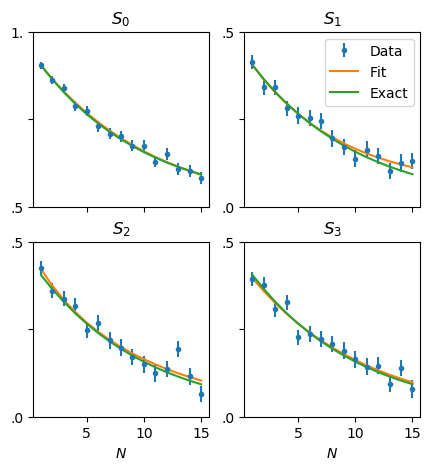}
    \caption{The predicted and measured character-weighted survival probability for a random error channel. The exact decay (green) is an exponential decay given by one of Eq. \ref{eq:fitFormMatchgates}. We estimate $S_i(N)$ by applying random gates and measuring the final state (blue points). The data is fit to an appropriate function (orange) from which we estimate the fidelity.}
    \label{fig:ExampleFitMatchgates}
\end{figure}
\begin{figure}
    \centering
    \includegraphics[width=\columnwidth]{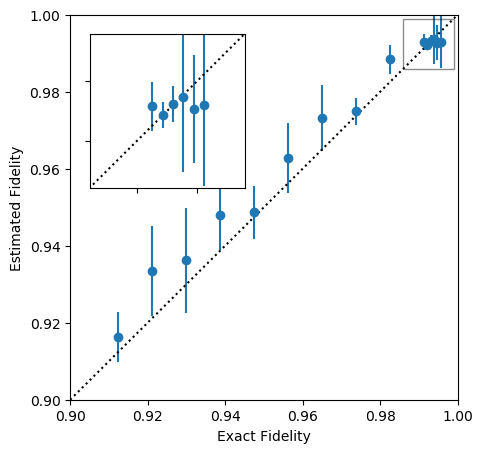}
    \caption{The exact and estimated fidelity for a selection of random errors. Each estimate was based on data taken over 15 different lengths $N$. Each estimate was arrived at by applying a total of 300,000 unitary group elements. The diagonal line denotes the points where the exact and estimated fidelities are equal. The data agree with the line with a reduced $\chi^2$ value of $1.0$, indicating good agreement.}
    \label{fig:FidelityEstimatesMatchgates}
\end{figure}

\section{Conclusion and Discussions}
\label{section:conclusion}
In this work, we extended the recently introduced character RB of \cite{helsen2019new} to groups with multiplicity. Compared to earlier work on benchmarking arbitrary groups \cite{brown2018randomized,francca2018approximate}, our method allows us to accurately determine the fidelity and fit fewer exponentials to experimental data. The generalization to non-multiplicity-free groups was essential to deriving a rigorous version of subspace RB and a scalable RB protocol for the matchgate group. This generalization also allowed us to develop an improved leakage RB protocol.

While we derived the character RB procedure in more generality than \cite{helsen2019new}, our generalization still requires groups of small multiplicity, since the multiplicity of the group determines the number of exponential decays in our fit function. Robustly fitting a sum of many exponential decays is challenging, especially when the decay rates are roughly equal \rev{\cite{bromage1983quantification,clayden1992multiexponential}}. It is likely straightforward to benchmark groups in which the trivial irrep has multiplicity three, as the corresponding decay $S_0(N)=A+B\lambda_{0,1}^N+C\lambda_{0,2}^N$ has only five real parameters. An irrep of multiplicity three with a real character function $\chi$ has a decay with six parameters, which may be feasible with sufficient data. A general irrep of multiplicity three, however, requires fitting nine real parameters, which is likely unfeasible for realistic amounts of data. Higher-multiplicity irreps are correspondingly more difficult. All of the groups we considered in the examples in this paper decomposed into irreps with multiplicity at most 2.

All our applications involved a group that preserved some subspace of the Hilbert space. In the case of subspace RB, the group preserved the triplet and singlet subspaces; in the case of leakage RB, the computational and leakage subspaces; and in the case of matchgate RB, the even and odd parity subspaces. Any group that preserves subspaces necessarily has multiplicity, since there is always a copy of the trivial irrep in each subspace. It is an open question whether non-multiplicity-free character RB has useful applications to groups that do not preserve subspaces but nonetheless have multiplicity.

One group related to the matchgate group that would be of immediate experimental interest is the \textbf{XY group}, the subgroup of the matchgate group generated by only nearest-neighbor XY mixers $U_{XY}(\theta)=\exp\left\{i\theta(X_1\otimes X_2+Y_1\otimes Y_2)\right\}$. Unlike general matchgates, XY mixers can be naturally realized on superconducting qubits \cite{abrams2019implementation,foxen2020demonstrating}, and they are a necessary ingredient in extensions of the QAOA algorithm \cite{hadfield2019quantum,wang2020x,cook2019quantum}. In addition, XY mixers are efficiently simulable on a line but become universal on nontrivial graphs, just like the full matchgate group\cite{brod2014computational}. However, XY mixers on $N$ qubits preserve the $(N+1)$ subspaces of definite Hamming weight; this implies that the trivial representation of the XY group must have multiplicity $(N+1)$. \rev{Thus, our method cannot be used to scalably benchmark the XY group; even $N=2$ qubits is likely infeasible. On the other hand, \cite{helsen2020matchgate} recently introduced a compilation of general two-qubit matchgates into products of four XY mixers and single-qubit gates. Using this decomposition, the average fidelity of the resulting two-qubit matchgates can be used as a proxy for the fidelity of the XY mixers. This method is similar to the benchmarking framework in our Sec \ref{section:subspace}, where we compile group elements into a fixed number of gates of interest (in our case, $U_{ZZ}$), with the modification that \cite{helsen2020matchgate} allows the gate of interest $XY(\theta)$ to vary. It is an open question if there is a generalization of this compilation to the matchgate group on $N>2$ qubits.}

While our leakage RB necessitates the fewest assumptions to date, it is still too restrictive for many experimental implementations. Most notably, our RB requires the set of gates to be a group, which may be unrealistic; often, the gates will only form a group modulo rotations in the leakage space. In experimental implementations of leakage RB, this problem is usually simply ignored and an exponential decay is posited to exist with the usual relation to the leakage rate \cite{chen2016measuring,andrews2019quantifying}. It is worth exploring whether the methods used here can be further extended to such sets of gates that are only groups in the computational subspace, modulo rotations in the leakage subspace, to provide a more rigorous foundation for leakage RB experiments.

There are two obvious directions for further applications of character RB, with or without multiplicity. First, character RB has the potential to drastically expand the family of groups that can be scalably benchmarked. This requires both finding a group $G$ that can be efficiently compiled into elementary gates whose multiplicity is bounded as the number of qubits $n$ increases, as well as finding a subgroup $\overline G\subseteq G$ whose irreps have slowly growing dimension. As a simple example, the subgroups of the Clifford group considered in \cite{brown2018randomized} likely have a scalable protocol based on character RB, with $\overline G$ given by the Pauli group. Increasing the number of groups that can be scalably benchmarked gives new ways of characterizing compiled gates, especially non-Clifford gates.

Second, character RB can be used to characterize specific elementary gates by combining these gates into a group, as we did in Section \ref{section:subspace} for subspace RB. This requires finding a group that can be implemented by combining a fixed number of the gate to be characterized with known high-fidelity gates. Constructing these groups is a non-trivial task, as we have seen in the case of the $U_{ZZ}$ operator above. We leave the exploration of such applications to future work. 

\rev{\emph{Note added.} After the first version of this paper was posted, \cite{helsen2020matchgate} was posted to the arXiv which also proposes a matchgate RB. Their method relies on enlarging the matchgate group with additional unitaries to avoid representations with multiplicity, but is otherwise similar to ours. As we mentioned in this paper, our character RB does not apply to the group generated by nearest-neighbor XY gates. While \cite{helsen2020matchgate} does not propose a method to benchmark the group generated by nearest-neighbor XY mixers, they do demonstrate a method to compile two-qubit matchgate elements using a fixed number of XY mixers and additional single-qubit gates, allowing the matchgate RB to be used to characterize XY mixers, as discussed above.}

\section*{Acknowledgements}

JC thanks Alexandre Pyvovarov for useful discussions on representations of $\bigwedge^i\mathbbm{C}^{2N}$. We are grateful for support from NASA Ames Research Center, the  
NASA Advanced Exploration systems (AES) program, and the NASA Transformative Aeronautic Concepts Program (TACP).
We are also grateful for support from the AFRL Information Directorate under grant F4HBKC4162G001.
JC was supported by the USRA Feynman Quantum Academy funded by the NAMS R\&D Student Program at NASA Ames Research Center. 
JC and ZW are also supported by NASA Academic Mission Services, Contract No. NNA16BD14C.

\bibliography{references.bib}

\appendix

\section{Gate-dependent errors}
\label{appendix:gateDependent}

In this appendix, we extend the work of \cite{wallman2018randomized,merkel2018randomized,helsen2019new} on gate-dependent errors to the case of non-multiplicity-free character RB. Ref. \cite{helsen2019new} had previously generalized \cite{wallman2018randomized} to establish that multiplicity-free character RB is robust to gate-dependent errors. Rather than follow the method of \cite{wallman2018randomized,helsen2019new} we use the Fourier transform method of \cite{merkel2018randomized}, which is more natural for groups with multiplicity. Our ultimate goal is the following theorem:
\begin{theorem}
    Let $G$ be a benchmarking group, and let $i$ be an irrep of the natural representation with multiplicity $a_i$. Assume each gate $U\in G$ is realized as a noisy operator $\eta(U)$, but \emph{do not} assume we can write $\eta(U)=\hat{\Lambda}\hat{U}$ for some $U$-independent noise channel $\Lambda$. Then the character-weighted survival probability is given by
    $$
    S_i(N) = \sum_{j=1}^{a_i}C_{i,j}\lambda_{i,j}^N+\epsilon_N
    $$
    where $\epsilon_N$ is an error term satisfying $|\epsilon_N|<\delta_1\delta_2^N$ and $\delta_1,\delta_2$ are both small for high-fidelity gates. Since we know that $\lambda_{i,j}\approx 1$ for high-fidelity gates, $\epsilon_N$ is negligible compared to $S_i(N)$ for moderately large $N$.
    \label{thm:gateDependence}
\end{theorem}

This theorem implies we may safely use the RB protocols even in the presence of gate-dependent errors, although we will see the interpretation of the estimated fidelity is slightly modified. 

In what follows, we will use the notation $\mathbbm{E}\left[\cdot\right]$ for the average $\frac{1}{|G|}\sum_{U\in G}\left(\cdot\right)$ or $\int_G dU\left(\cdot \right)$ to make our equations cleaner. We will also use the shorthand $d_i$ for $\dim(\mathcal{H}_i)$.

\subsection{The generalized Fourier transform and its application to character RB}

We first define a generalization of the Fourier transform to matrix-valued functions of a group $G$\cite{moore2015approximate,gowers2017inverse}. For any group $G$ we define $\widetilde G$ to index the irreps of $G$, and we assume WLOG that the irreps are unitary. Given a function $\eta: G\rightarrow \mathcal{L}(\mathbbm{C}^D)$, for each $i\in \widetilde G$ we define the {\bf Fourier transform} $\widetilde\eta(i) \in\mathcal{L}(\mathbbm{C}^D)\otimes\mathcal{L}(\mathcal{H}_i)$ to be
\begin{equation}
\widetilde\eta(i) = \mathbbm{E}\left[\eta(U)\otimes \phi_i^*(U)\right].
\label{eq:FourierDefn}
\end{equation}
where $\phi_i:G\rightarrow \mathcal{L}(\mathcal{H}_i)$ is the $i$th irrep.

Given two matrix-valued functions $\eta,\xi:G\rightarrow \mathcal{L}(\mathbbm{C}^D)$, we can also define the {\bf convolution} $(\eta *\xi)$ by
\begin{equation}
    (\eta*\xi)(U_0)=\mathbbm{E}\left[\eta(U^\dagger)\xi(U U_0)\right].
    \label{eq:ConvolutionDefn}
\end{equation}

The generalized Fourier transform shares many properties with the usual Fourier transform; in particular, we will use the following identities \cite{gowers2017inverse,merkel2018randomized}:
\begin{align}
    (\widetilde{\eta * \xi})(i) &= \widetilde\eta(i)\widetilde\xi(i)\label{eq:ConvolutionIdentity} \\
    \mathbbm{E}\left[\text{Tr}\left(\eta(U)\xi^\dagger(U)\right)\right] &= \sum_i d_i\text{Tr}\left(\widetilde\eta(i)\widetilde\xi^\dagger(i)\right)\label{eq:ParsevalIdentity}\\
    \eta(U) &= \sum_i d_i \text{Tr}_i\left([\mathbbm{1}\otimes \phi_i^T(U)]\widetilde{\eta}(i)\right) \label{eq:InverseIdentity}
\end{align}
where in the last line, $\text{Tr}_i\left(\cdot\right)$ is the partial trace over $\mathcal{H}_i$. Eq. \ref{eq:ConvolutionIdentity} is the analogue of the usual convolution identity for Fourier transforms, Eq. \ref{eq:ParsevalIdentity} is the analogue of Parseval's identity, and Eq. \ref{eq:InverseIdentity} gives the inverse Fourier transform.

The generalized Fourier transformation is useful because it allows us to express the result of a character RB experiment in a simpler form. A character RB experiment estimates a matrix element of the operator
$$
    \hat{O}_i:=\mathbbm{E}\left[\eta(U_1^\dagger\cdots U_N^\dagger)\eta(U_N)\cdots \eta(U_2)\eta(U_1U_0)\chi_{\overline i}^*(U_0)\right]
$$
where the expectation value is over all $U_0\in \overline{G}$, $U_1,...,U_N\in G$. Through the change of variables $U_i\rightarrow U_iU_{i-1}\cdots U_1$ for $i=1,...,N$, we can rewrite this expression as a convolution:
\begin{align*}
    \hat{O}_i&=\mathbbm{E}\left[\eta(U_N^\dagger)\eta(U_N U_{N-1}^\dagger)\cdots \eta(U_2U_1^\dagger)\eta(U_1U_0)\chi_{\overline i}^*(U_0)\right]\\
    &=\mathbbm{E}\left[\underbrace{(\eta *\cdots *\eta)}_{(N+1) \text{ times}}(U_0)\chi_{\overline i}^*(U_0)\right]
\end{align*}
Using the inverse Fourier transform (Eq. \ref{eq:InverseIdentity}) we can write $(\eta *\cdots*\eta)(U_0)$ in terms of $(\widetilde{\eta *\cdots*\eta})(i')$, while the convolution identity (Eq. \ref{eq:ConvolutionIdentity}) allows us to simplify $(\widetilde{\eta *\cdots*\eta})(i')=\widetilde{\eta}(i')^{N+1}$. In total, we find
\begin{align*}
    \hat{O}_i &= \sum_{i'}d_{i'}\text{Tr}_{i'}\left(
    \left[\mathbbm{1}\otimes\mathbbm{E}\left[\chi_{\overline i}^*(U_0)\phi_{i'}(U_0)\right]^T\right]\widetilde\eta(i')^{N+1}\right).
\end{align*}
We now use the projection formula (Fact \ref{fact:Projection}) to note that $d_{\overline i}\mathbbm{E}\left[\chi_{\overline i}^*(U_0)\phi_{i'}(U_0)\right]$ is just the projection of $\phi_{i'}$ onto the irrep $\overline{i}$ of $\overline G$. By assumption, the irrep $\phi_{\overline i}$ is a subrepresentation of only $\phi_i$, and not a subrepresentation of any $\phi_{i'}$ with $i'\neq i$. Therefore, 
\begin{align*}
    \hat{O}_i &= \frac{d_i}{d_{\overline i}}\text{Tr}_{i}\left(
    \left[\mathbbm{1}\otimes \hat{P}_{\overline i}^T\right]\widetilde\eta(i)^{N+1}\right).
\end{align*}

We therefore see that the outcome of a character RB experiment, $S_i(N)$, can be described by the Fourier transform of $\eta$ via
\begin{align}
S_i(N) &= \langle\langle M_i|\hat{\Lambda}_M \hat{O}_i\hat{\Lambda}_P|\rho_i\rangle\rangle     \label{eq:decayFormFourier}\\
&=\frac{d_i}{d_{\overline i}}\langle\langle M_i|\hat{\Lambda}_M \text{Tr}_{i}\left(
    \left[\mathbbm{1}\otimes \hat{P}^T_{\overline i}\right]\widetilde\eta(i)^{N+1}\right)\hat{\Lambda}_P|\rho_i\rangle\rangle \nonumber
\end{align}
and the decay of $S_i(N)$ is determined by the eigenvalues of $\widetilde{\eta}(i)$.

\subsection{Simplifying the decay}

In the case of ideal gates $\eta_{\text{ideal}}(U)=\hat{U}$, we have that $\widetilde{\eta}_{\text{ideal}}(i)$ is given by
$$
\widetilde{\eta}_{\text{ideal}}(i) = \mathbbm{E}\left[\hat{U}\otimes \phi_i(U)\right]
$$
This can be simplified by noting that the map $\eta_{\text{ideal}}\otimes \phi_i:U\mapsto \hat{U}\otimes \phi_i(U)$ is a representation of $G$, and $\mathbbm{E}\left[\hat{U}\otimes \phi_i(U)\right]$ is the projection of this representation onto the copies of the trivial irrep (Fact \ref{fact:Projection}). We can count the multiplicity of the trivial irrep in $(\eta_{\text{ideal}}\otimes\phi_i)$ using the following fact:

\begin{fact}[Schur orthonormality]
If $\chi$ is the character of an arbitrary representation $\phi$, and $\chi_i$ is the character of an irrep $\phi_i$, the multiplicity $a_i$ of $\phi_i$ is
$$
a_i=\frac{1}{|G|}\sum_{U\in G} \chi_i^*(U)\chi(U).
$$
\label{fact:orthogonality}
\end{fact}\noindent
For a proof, see \cite{fulton2013representation}.

Since the trivial irrep has $\chi_i(U)=1$, we have that the multiplicity of the trivial irrep in $(\eta_{\text{ideal}}\otimes \phi_i)$ is given by
$$
\mathbbm{E}\left[\text{Tr}\left(\hat{U}\otimes \phi_i^*(U)\right)\right]=\mathbbm{E}\left[\chi_i^*(U)\text{Tr}\left(\hat U\right)\right]=a_i.
$$
In other words, $\tilde\eta_{\text{ideal}}(i)$ is a rank-$a_i$ projector.

We can explicitly find the form of $\tilde\eta_{\text{ideal}}(i)$ by constructing $a_i$ trivial irreps of $(\eta_{\text{ideal}}\otimes \phi_i)$. Let $\{|\psi^{i}_n\rangle\rangle\}$ be an orthonormal basis for $\mathcal{H}_i$, and let $\{|\psi^{i,j}_n\rangle\rangle\}$ be the corresponding basis for the $j$th copy of $\mathcal{H}_i$ inside $\mathcal{H}\otimes\mathcal{H}$. It is straightforward to show that
\begin{equation*}
|\Psi^{i,j}\rangle\rangle := \frac{1}{\sqrt{d_i}}\sum_{n=1}^{d_i}|\psi_n^{i,j}\rangle\rangle\otimes|\psi_n^i\rangle\rangle
\end{equation*}
spans an irrep for each $j=1,...,a_i$. Therefore,
\begin{equation}
    \tilde\eta_{\text{ideal}}(i)=\sum_{j=1}^{a_i}|\Psi^{i,j}\rangle\rangle\langle\langle\Psi^{i,j}|
    \label{eq:idealEta}
\end{equation}

A realistic experiment will have gates described by a function $\eta(U)$ that is some small perturbation from $\eta_{\text{ideal}}(U)$. Perturbing $\eta_{\text{ideal}}(U)$ by a small amount will perturb $\tilde\eta_{\text{ideal}}(i)$ by a small amount, since the Fourier transform is a linear operation. Thus $\tilde\eta(i)$ is a perturbation of a rank-$a_i$ projector for high-fidelity gates, so that $\tilde{\eta}(i)$ has $a_i$ eigenvalues close to $1$, which we will denote by $\lambda_{i,j}$, and the remaining eigenvalues close to $0$. This is sufficient to make $S_i(N)$ dominanted by $a_i$ exponential decays, corresponding to the $a_i$ largest eigenvalues (see Eq. \ref{eq:decayFormFourier}). This proves Thm. \ref{thm:gateDependence}.

\subsection{Computing the average fidelity}
If we define $\eta(U)=\hat\Lambda_U \hat{U}$, with $\Lambda_U$ the gate-dependent error channel, then we can define an average fidelity
\begin{equation}
F_{\text{av}} = \frac{\mathbbm{E}\left[\text{Tr}(\hat\Lambda_U)\right]+d}{d^2+d}
\label{eq:AverageFidelityGateDependent}
\end{equation}
Comparing to Eq. \ref{eq:TraceFidelity}, we see that this is simply the average of the individual fidelities $F_{\Lambda_U}$.

We can express $F_{\text{av}}$ in terms of the $a_i$ largest eigenvalues of $\tilde\eta(i)$ as follows. We first note that we may write
\begin{align*}
\mathbbm{E}\left[\text{Tr}\left(\hat\Lambda_U\right)\right]&=\mathbbm{E}\left[\text{Tr}\left(\eta(U)\eta^\dagger_{\text{ideal}}(U)\right)\right]\\
&=\sum_i d_i\text{Tr}\left(\tilde\eta(i)\tilde\eta^\dagger_{\text{ideal}}(i)\right)\\
&=\sum_{i=1}^I\sum_{j=1}^{a_i}d_i\langle\langle\Psi^{i,j}|\tilde\eta(i)|\Psi^{i,j}\rangle\rangle
\end{align*}
where in the second line we used the Parseval identity (Eq. \ref{eq:ParsevalIdentity}) to move to Fourier space, and in the third line we used the explicit form of $\tilde\eta_{\text{ideal}}(i)$ (Eq. \ref{eq:idealEta}). To first order in $\left(\tilde\eta(i)-\tilde\eta_{\text{ideal}}(i)\right)$, we have that
$$
\sum_{j=1}^{a_i}\langle\langle\Psi^{i,j}|\tilde\eta(i)|\Psi^{i,j}\rangle\rangle \approx \sum_{j=1}^{a_i}\lambda_{i,j}
$$
Therefore, we can rewrite Eq. \ref{eq:AverageFidelityGateDependent} as
$$
F_{\text{av}} \approx \frac{\sum_{i=1}^I d_i\sum_{j=1}^{a_i}\lambda_{i,j}+d}{d^2+d}
$$
which is the same form as Eq. \ref{eq:fidelityEst} in the case of gate-independent noise.


\section{The generalized Clifford group is a unitary 2-design}
\label{appendix:WeylGroup}

In this Appendix, we prove the generalized Clifford group considered in Section \ref{section:irrepsSubspace} is a unitary 2-design. We will give a fully general treatment for arbitrary sets of $n$ qudits with $d>2$ prime, although we need only the case of $n=1$, $d=3$ for our subspace benchmarking above. This result can be inferred from results proven in \cite{chau2005unconditionally}, but we give a direct proof below. We first review the construction of the generalized Clifford groups as introduced in \cite{hostens2005stabilizer}.

For a $d$-level system, define analogues of the $X$ and $Z$ qubit operators \cite{gottesman1998fault}:
$$
X|z\rangle = |z+1\rangle \qquad Z|z\rangle = \omega^z|z\rangle
$$
where $\omega := e^{2\pi i/d}$ and addition is performed modulo $d$. These generalized $X$ and $Z$ operators are unitary and satisfy $ZX = \omega XZ$.

For a set of $n$ qudits, define the $d$-dimensional generalization of the Pauli group as (this only holds for $d$ odd; see \cite{hostens2005stabilizer} for the definition for $d$ even):
$$\mathcal{P}:=\{\omega^\eta X_1^{a_1} Z_1^{b_1}\cdots X_n^{a_n} Z_n^{b_n}:\eta,a_i,b_i\in \mathbbm{Z}_d\}.$$
We will write a general element of the Pauli group as 
$$
\omega^\eta X_1^{a_1} Z_1^{b_1}\cdots X_n^{a_n} Z_n^{b_n}:=\omega^\eta \XZ(\vec v),\quad \vec{v}:=\left(\begin{smallmatrix}\vec a\\\vec b\end{smallmatrix}\right).
$$
Multiplication of general elements of the Pauli group is given by
$$
\XZ(\vec{v})\XZ(\vec{w}) = \omega^{\vec{v}^TQ\vec{w}}\XZ(\vec{v}+\vec{w})
$$
where $Q$ is defined by $Q = \left(\begin{smallmatrix}0&0\\\mathbbm{1} & 0\end{smallmatrix}\right)$. This demonstrates that $\mathcal{P}$ is indeed a group.

The generalized Clifford group is defined to be the set of all unitaries that stabilize $\mathcal{P}$:
$$
G = \{U :U\mathcal{P} U^\dagger = \mathcal{P}\}.
$$
An element $U\in G$ is defined (up to a global phase) by its action on $X_i$ and $Z_i$. We define the matrix $M$ and vector $\vec{h}$ such that for each unit vector $\hat{e}_i\in \mathbbm{Z}_d^{2d}$ we have
$$
U \XZ(\hat{e}_i) U^\dagger = \omega^{h_i}\XZ(M\hat{e}_i)
$$
It then follows that a general element $\XZ(a)$ is transformed as
\begin{equation}
\label{eq:generalConjugate}
\begin{array}{c}
U\XZ(\vec{v})U^\dagger = \omega^\eta \XZ(M\vec{v})\\
\eta := \left(\vec{h}-\frac{\text{diag}(M^T Q M)}{2}\right)^T\vec{v}+\vec{v}^T\left(M^T Q M-Q\right)\frac{\vec{v}}{2}
\end{array}
\end{equation}

Not every matrix $M$ can be realized by a unitary operator. To derive a restriction on $M$, we consider the commutation relation (where we define $P=Q-Q^T$):
\begin{align*}
\XZ(\vec{v})\XZ(\vec{w})&=\omega^{\vec{v}^TP\vec{w}}\XZ(\vec{w})\XZ(\vec{v})\\
U\XZ(\vec{v})\XZ(\vec{w})U^\dagger&=\omega^{\vec{v}^TP\vec{w}}U\XZ(\vec{w})\XZ(\vec{v})U^\dagger\\
\XZ(M\vec v)\XZ(M\vec w)&=\omega^{\vec{v}^TP\vec{w}}\XZ(M\vec w)\XZ(M\vec v)\\
\omega^{\vec{v}^TM^TPM\vec{w}}\XZ(M\vec w)\XZ(M\vec v)&=\omega^{\vec{v}^TP\vec{w}}\XZ(M\vec w)\XZ(M\vec v)
\end{align*}
where we have ignored phase factors common to both sides. We see that we must have $P=M^TP M$; such an $M$ is called a {\bf symplectic matrix}. This is the only restriction on $M,h$, as \cite{hostens2005stabilizer} demonstrated how to explicitly construct unitaries to implement any $M,h$ provided $M$ is symplectic.

To prove $G$ forms a unitary $2$-design, we need to show (see Section \ref{section:irrepsSubspace} of the main text) 
$$
\frac{1}{|G|}\sum_{U\in G} p(U,U^*)=\int dU\ p(U,U^*)
$$
for any balanced polynomial $p(U,U^*)$ of degree at most $2$ in the elements of $U$ and $U^*$. Any such $p(U,U^*)$ can be written as a linear combination of terms of the form $UAU^\dagger B U C U^\dagger$ and $UDU^\dagger$, where $A,B,C,D$ are matrices. We are thus reduced to proving
\begin{equation}
\frac{1}{|G|}\sum_{U\in G} UAU^\dagger B U C U^\dagger=\int dU\ UAU^\dagger B U C U^\dagger
\label{eq:firstCondition}
\end{equation}
\begin{equation}
\frac{1}{|G|}\sum_{U\in G}  U D U^\dagger=\int dU\ UDU^\dagger
\label{eq:secondCondition}
\end{equation}
for arbitrary matrices $A,B,C,D$.

In the following, we will make repeated use of an elementary identity of complex roots of unity.
\begin{fact}
If $\vec{w}\in \mathbbm{Z}_d^{2n}\setminus\{0\}$ is any nonzero vector, then $$\sum_{\vec{v}}\omega^{\vec{v}^T\vec{w}}=0.$$
\label{fact:identity}
\end{fact}
\subsection{Degree 1 polynomials}

Let's start by proving Eq. \ref{eq:secondCondition}. Without loss of generality, we can assume $D=\XZ(\vec{v})$, since such matrices form a basis. The RHS of this equation is invariant under conjugation by arbitrary unitaries; thus, it must be proportional to the identity matrix. Noting that $\Tr(\text{RHS})=\Tr(D)$ and that $\Tr\left[\XZ(\vec{v})\right]=0$ whenever $\vec{v}\neq 0$, we find
\begin{align*}
    \text{RHS} &= \left\{\begin{array}{ll}\mathbbm{1}, & \vec{v}=0\\0, &\text{else}.\end{array}\right.
\end{align*}
We evaluate the LHS by using Eq. \ref{eq:generalConjugate} for the conjugation of a general Pauli element:
\begin{align*}
    \text{LHS}&=\frac{1}{|G|}\sum_{U\in G} U\XZ(\vec{v}) U^\dagger\\
    &=\frac{1}{|G|}\sum_{\substack{M,\vec{h}\\M^TPM=P}}\omega^{\eta}\XZ(M\vec{v})
\end{align*}
We note that $\eta = \vec{h}^T\vec{v}+(\cdots)$, where $(\cdots)$ denotes terms that do not depend on $\vec{h}$. We see by Fact \ref{fact:identity} that for fixed $M$ the sum over $\vec h$ gives zero unless $\vec{v}=0$, while when $\vec{v}=0$ it is clear $\text{LHS}=\mathbbm{1}$. This proves Eq. \ref{eq:secondCondition}.

\subsection{Degree 2 polynomials}

We now turn to Eq. \ref{eq:firstCondition}. We prove this using methods from \cite{dankert2009exact}, who proved the case $d=2$. First, we note that the RHS of Eq. \ref{eq:firstCondition} is covariant in $B$: sending $B\rightarrow UBU^\dagger$ sends $\text{RHS}\rightarrow U(\text{RHS})U^\dagger$ for any unitary $U$. The only covariant linear functions of $B$ are $\frac{\Tr(B)\mathbbm{1}}{d^n}$ and $\left[B-\frac{\Tr(B)\mathbbm 1}{d^n}\right]$, so the RHS must be of the form \cite{emerson2005scalable}
\begin{equation}
\text{RHS}=q\left[B-\frac{\Tr(B)\mathbbm 1}{d^n}\right]+p\frac{\Tr(B)\mathbbm 1}{d^n}.
\label{eq:overallForm}
\end{equation}
To determine $p$ we plug in $B=\mathbbm{1}$ and note that  
$$
\text{RHS}=\int dU\ UA C U^\dagger = \frac{\Tr(AC)}{d^n}\mathbbm{1},
$$
while simultaneously according to Eq. \ref{eq:overallForm},
$$
\text{RHS} = p\mathbbm{1}
$$
so $p=\frac{\Tr(AC)}{d^n}$. To determine $q$, we consider plugging in $B=|i\rangle\langle j|$. Denoting the result when plugging in $B=|i\rangle\langle j|$ as $(\text{RHS})_{ij}$, we can evaluate
\begin{align*}
    \sum_{i,j}\langle i|(\text{RHS})_{ij}|j\rangle &= \sum_{i,j}\int dU\ \langle i|UAU^\dagger|i\rangle\langle j|UCU^\dagger|j\rangle\\
    &=\Tr(A)\Tr(C).
\end{align*}
On the other hand, Eq. \ref{eq:overallForm} gives
$$
\sum_{i,j}\langle i|(\text{RHS})_{ij}|j\rangle =(d^{2n}-1)q+p
$$
so $q = \frac{d^n\Tr(A)\Tr(C)-\Tr(AC)}{d^n(d^{2n}-1)}$. Thus in total, we have

\begin{multline}
\text{RHS}=\frac{d^n\Tr(A)\Tr(C)-\Tr(AC)}{d^n(d^{2n}-1)}\left[B-\frac{\Tr(B)\mathbbm 1}{d^n}\right]\\
+\frac{\Tr(AC)\Tr(B)\mathbbm 1}{d^{2n}}.
\label{eq:FinalRHSForm}
\end{multline}

Without loss of generality, we can specialize to the case where $A=\XZ(\vec{v}_A)$, $B=\XZ(\vec{v}_B)$, and $C=\XZ(\vec{v}_C)$, whence Eq. \ref{eq:FinalRHSForm} gives 
$$
\text{RHS}=\left\{\begin{array}{ll}
\XZ(\vec{v}_B), & \vec{v}_A=\vec{v}_C=0\\
\omega^{-\vec{v}_A^T Q\vec{v}_A}\mathbbm{1}, & \vec{v}_A=-\vec{v}_C\neq 0,\ \vec{v}_B=0\\
-\frac{\omega^{-\vec{v}_A^T Q\vec{v}_A}}{d^{2n}-1}\XZ(\vec{v}_B), & \vec{v}_A=-\vec{v}_C\neq 0,\ \vec{v}_B\neq 0\\
0, & \text{else}.
\end{array}\right.
$$

We now need to evaluate the LHS of Eq. \ref{eq:firstCondition} for each of the four cases above. In the first case, we find
\begin{align*}
    \text{LHS} = \frac{1}{|G|}\sum_{U\in G} \XZ(\vec{v}_B)=\XZ(\vec{v}_B)
\end{align*}

In the second case, we use Eq. \ref{eq:generalConjugate} to simplify each summand in the LHS
\begin{align*}
    U&\XZ(\vec{v}_A)U^\dagger U\XZ(\vec{v}_C)U^\dagger\\
    &=\omega^{\eta_A+\eta_C} \XZ(M\vec{v}_A)\XZ(M\vec{v}_C)\\
    &=\omega^{\eta_A+\eta_C+\vec{v}_A^TM^TQM\vec{v}_C} \mathbbm{1}\\
    &=\omega^{\vec{v}_A^T\left(M^T Q M-Q\right)\vec{v}_A-\vec{v}_A^TM^TQM\vec{v}_A} \mathbbm{1}\\
    &=\omega^{-\vec{v}_A^T Q \vec{v}_A} \mathbbm{1}.
\end{align*}
Therefore, the average over the group $G$ gives $\omega^{-\vec{v}_A^T Q \vec{v}_A} \mathbbm{1}$.

In the third case, we again simplify each summand using Eq. \ref{eq:generalConjugate}, but with an additional $B$ in between:
\begin{align*}
&U\XZ(\vec{v}_A)U^\dagger \XZ(\vec{v}_B)U\XZ(\vec{v}_C)U^\dagger\\
&\ =\omega^{\eta_A+\eta_C} \XZ(M\vec{v}_A)\XZ(\vec{v}_B)\XZ(M\vec{v}_C)\\
&\ =\omega^{\eta_A+\eta_C+\vec{v}_A^TM^TQ\vec{v}_B-\vec{v}_B^TQM\vec{v}_A-\vec{v}_A^TM^TQM\vec{v}_A}\XZ(\vec{v}_B)\\
&\ =\omega^{\vec{v}_A^TM^TP\vec{v}_B-\vec{v}_A^TQ\vec{v}_A}\XZ(\vec{v}_B).
\end{align*}
The average over $\vec{h}$ does not affect this sum, so we only need to consider the average over $M$. We evaluate the average by realizing that if $d$ is prime, the Clifford group sends every non-identity Pauli string to every other non-identity Pauli string uniformly. Thus, letting $M$ run over all symplectic matrices makes $M\vec{v}_A$ run uniformly over all vectors $M\vec{v}_A\in \mathbbm{Z}_d^{2n}\setminus \{0\}$. Therefore, the LHS is given by
\begin{align*}
\text{LHS}&=\frac{1}{d^{2n-1}}\sum_{\vec{v}\neq 0}\omega^{\vec{v}^T P \vec{v}_B-\vec{v}_A^TQ\vec{v}_A}\XZ(\vec{v}_B)\\
&=-\frac{\omega^{-\vec{v}_A^TQ\vec{v}_A}}{d^{2n-1}}\XZ(\vec{v}_B)\left[1-\sum_{\vec{v}}\omega^{\vec{v}^T P \vec{v}_B}\right]\\
&=-\frac{\omega^{-\vec{v}_A^TQ\vec{v}_A}}{d^{2n-1}}\XZ(\vec{v}_B)
\end{align*}
where in the final step, we used Fact \ref{fact:identity}.

In the last case, we have that each summand is of the form
\begin{align*}
&U\XZ(\vec{v}_A)U^\dagger \XZ(\vec{v}_B)U\XZ(\vec{v}_C)U^\dagger\\
&\ =\omega^{\eta_A+\eta_C+\vec{v}_A^TM^TQ\vec{v}_B-\vec{v}_B^TQM\vec{v}_A-\vec{v}_A^TM^TQM\vec{v}_A}\XZ(\vec{v}_B)\\
&\ =\omega^{\vec{h}^T(\vec{v}_A+\vec{v}_C)+(\cdots)}\XZ(\vec{v}_B)
\end{align*}
where $(\cdots)$ represents terms that are independent of $\vec h$. We can again apply Fact \ref{fact:identity} to find that the sum over $\vec{h}$ gives zero. We have thus proved $\text{LHS}=\text{RHS}$ for each of the four cases, which establishes Eq. \ref{eq:firstCondition}.

\section{Leakage RB irreps}
\label{appendix:choosingLeakageGroup}

Let $G$ be a unitary group \rev{indexed} by \rev{$b\in B$},
\rev{\begin{align*}G&=\{U_{b,\sigma}:b\in B\ \sigma=\pm 1\}\\&=\{U_{1,b}\oplus \sigma U_{2,b}:b\in B,\ \sigma=\pm 1\},\end{align*}}\noindent
where \rev{$G_1=\{U_{1,b}:b\in B\}$} and \rev{$G_1=\{U_{2,b}:b\in B\}$} are each unitary 1-designs on their respective subspaces. \rev{First, we prove} that $|\mathbbm{1}_1\rangle\rangle$ and $|\mathbbm{1}_2\rangle\rangle$ are the only trivial irreps of the natural representation of $G$. Next, we prove that if \rev{$G_1$} and \rev{$G_2$} are in addition unitary 2-designs and $d_1\neq d_2$ then $\mathcal{H}_{1\perp}$ is irreducible and multiplicity-free.

We start with the trivial irreps. It is clear that both $|\mathbbm{1}_1\rangle\rangle$ and $|\mathbbm{1}_2\rangle\rangle$ are trivial irreps. The trivial irrep has $\chi_0(U)=1$, so Fact \ref{fact:orthogonality} gives
\rev{\begin{align*}
a_0&=\frac{1}{|G|}\sum_{U\in G} \chi(U)\\
&=\frac{1}{2|B|}\sum_{\substack{b\in B\\\sigma=\pm}} \Tr(U_{b,\sigma}\otimes U_{b,\sigma}^*)\\
&=\frac{1}{2|B|}\sum_{\substack{b\in B\\\sigma=\pm}}\left[\begin{array}{l} \Tr(U_{1,b}\otimes U_{1,b}^*)+\sigma\Tr(U_{1,b}\otimes U_{2,b}^*)\\+\sigma\Tr(U_{2,b}\otimes U_{1,b}^*)+\Tr(U_{2,b}\otimes U_{2,b}^*)\end{array}\right]\\
&=\frac{1}{|B|}\sum_{b\in B} \left[\Tr(U_{1,b}\otimes U_{1,b}^*)+\Tr(U_{2,b}\otimes U_{2,b}^*)\right]\\
&=\int dU_1 \Tr(U_{1}\otimes U_{1}^*)+\int dU_2\Tr(U_{2}\otimes U_{2}^*)
\end{align*}}\noindent
where in the last line we used the fact that \rev{$G_1$} and \rev{$G_2$} are unitary 1-designs. These integrals just give the number of trivial irreps of the full unitary group on $\mathcal{H}_1$ and $\mathcal{H}_2$, respectively, which are known to be $1$. Thus, there are only two trivial irreps of the full unitary group.

Now, we consider $\mathcal{H}_{1\perp}$. First, we show $\mathcal{H}_{1\perp}$ is irreducible by using Fact \ref{fact:magnitudeCharacter}. Noting \rev{$\chi_{1,\perp}(U_{b,\pm})=\left(|\Tr(U_{1,b})|^2-1\right)$}, we have
\rev{\begin{align*}
\frac{1}{|G|}\sum_{U\in G}|\chi_{1\perp}(U)|^2&=\frac{1}{2|B|}\sum_{\substack{b\in B\\\sigma=\pm}}\left(|\Tr(U_{1,b})|^2-1\right)^2\\
&=\frac{1}{|B|}\sum_{b\in B}\left(|\Tr(U_{1,b})|^2-1\right)^2\\
&=\int dU_1\ \left(|\Tr(U_1)|^2-1\right)^2 \\
&= 1
\end{align*}}\noindent
where the third equality follows from the unitary 2-design property, and the fourth follows from the fact that $\mathcal{H}_{1\perp}$ is an irrep of the natural representation of the full unitary group on $\mathcal{H}_1\otimes\mathcal{H}_1$. Thus, we have $\mathcal{H}_{1\perp}$ irreducible.

To finish, we must prove that no other irrep of the natural representation is isomorphic to $\mathcal{H}_{1\perp}$. Every irrep of the natural representation is a subrepresentation of $\mathcal{H}_1\otimes \mathcal{H}_1$, $\mathcal{H}_1\otimes \mathcal{H}_2$, $\mathcal{H}_2\otimes \mathcal{H}_1$, or $\mathcal{H}_2\otimes \mathcal{H}_2$, since these subspaces are all invariant under the action of $G$. We know that the decomposition of $\mathcal{H}_1\otimes\mathcal{H}_1$ into irreps is $\mathcal{H}_1\otimes\mathcal{H}_1\simeq \mathcal{H}_{10}\otimes\mathcal{H}_{1\perp}$, by our work above, and thus no irreps in $\mathcal{H}_1\otimes\mathcal{H}_1$ can be isomorphic to $\mathcal{H}_{1\perp}$ besides $\mathcal{H}_{1\perp}$ itself.  Similarly, we know that the decomposition of $\mathcal{H}_2\otimes\mathcal{H}_2$ into irreps is $\mathcal{H}_2\otimes\mathcal{H}_2\simeq \mathcal{H}_{20}\otimes\mathcal{H}_{2\perp}$. We can ensure $\mathcal{H}_{1\perp}\not\simeq\mathcal{H}_{2\perp}$ by requiring $d_1\neq d_2$, as in the main text. We then have that no isomorphic representation exists in $\mathcal{H}_2\otimes\mathcal{H}_2$. For $\mathcal{H}_1\otimes\mathcal{H}_2$, and similarly for $\mathcal{H}_2\otimes\mathcal{H}_1$, we note that the character of the subrepresentation $\mathcal{H}_1\otimes\mathcal{H}_2$ is given by \rev{$\chi_{12}(U_{b,\sigma})=\sigma\Tr(U_{1,b})\Tr(U_{2,b})^*$}, and use Fact \ref{fact:orthogonality}:
\rev{\begin{align*}
    \frac{1}{|G|}&\sum_{U\in G} \chi_{1\perp}^*(U)\chi_{12}(U) \\
    &= \frac{1}{2|B|}\sum_{\substack{b\in B\\\sigma=\pm}} \sigma(|\Tr(U_{1,b_1})|^2-1)\Tr(U_{1,b})\Tr(U_{2,b})^*\\
    &=0
\end{align*}}\noindent
which shows that $\mathcal{H}_{1\perp}$ is an irrep with multiplicity $1$.

Note that we could also consider a group \rev{$$G'=\{U_{b,\phi}:b\in B\}=\{U_{1,b}\oplus(e^{i\phi}U_{2,b}):b\in B\}$$}\noindent with an arbitrary phase between subspaces $1$ and $2$ rather than simply a $\pm 1$ phase; the proof is identical. Many experimental platforms can easily implement a random phase between two subspaces, especially if the leakage subspace is at a different energy than the computational subspace, making this group potentially easier to sample from. We can also still compute $F_{\Lambda,1}$ with $\{U_{2,a}\}$ only a unitary $1$-design, provided $\mathcal{H}_2\otimes \mathcal{H}_2$ does not contain an irrep isomorphic to $\mathcal{H}_{1\perp}$. Finally, in the case that $d_1 = d_2$, we can instead simply require that there exists some \rev{$b\in B$} such that \rev{$|\Tr(U_{1,b})|^2\neq |\Tr(U_{2,b})|^2$}, a much weaker condition that still suffices to ensure $\mathcal{H}_{1\perp}\not\simeq\mathcal{H}_{2\perp}$. 

\end{document}